\documentclass[a4paper,11pt]{article}
\usepackage{jheppub}
\usepackage[utf8]{inputenc}

\usepackage{lmodern}

\usepackage{amsfonts}
\usepackage{amssymb}
\usepackage{mathrsfs}
\usepackage{mathtools}
\usepackage{mleftright}
\usepackage{booktabs}
\usepackage{tabulary}

\renewcommand{\Im}{\operatorname{Im}}

\def\i{\mathrm{i}}
\def\d{\mathrm{d}}
\def\g{\epsilon}
\def\m{\mathfrak{m}}
\def\vf{\mathrm{dvol}_\Sigma}

\DeclareMathOperator{\tr}{tr}
\DeclareMathOperator{\sign}{sign}
\DeclareMathOperator{\vol}{vol}
\DeclareMathOperator{\erf}{erf}
\DeclareMathOperator{\erfc}{erfc}

\DeclareMathOperator{\pv}{pv}
\DeclareMathOperator*{\Clim}{Clim}

\usepackage{tikz}
\usepackage{tikz-cd}

\title{\huge\boldmath Localization and resummation\\of unstable instantons in 2d Yang--Mills}

\author[a]{Luca Griguolo,}
\author[b]{Rodolfo Panerai,}
\author[c,d]{Jacopo Papalini,}
\author[e]{Domenico Seminara,}
\author[f]{\\and Itamar Yaakov}

\affiliation[a]{Dipartimento SMFI, Universit\`a di Parma and INFN Gruppo Collegato di Parma\\Viale G.P.\ Usberti 7/A, 43100 Parma, Italy}
\affiliation[b]{Institute for Theoretical Physics, University of Cologne\\Z\"ulpicher Stra{\ss}e 77, 50937 K\"oln, Germany}
\affiliation[c]{Galileo Galilei Institute for Theoretical Physics, INFN\\Largo Enrico Fermi 2, 50125 Firenze,
Italy}
\affiliation[d]{Department of Physics and Astronomy, Ghent University\\Krijgslaan, 281-S9, 9000 Gent, Belgium}
\affiliation[e]{Dipartimento di Fisica, Universit\`a di Firenze and INFN Sezione di Firenze\\Via G.\ Sansone 1, 50019 Sesto Fiorentino, Italy} 
\affiliation[f]{STAG Research Center \& Mathematical Sciences, University of Southampton\\Highfield, Southampton SO17 1BJ, UK}

\emailAdd{luca.griguolo@unipr.it}
\emailAdd{rpanerai@uni-koeln.de}
\emailAdd{jacopo.papalini@ugent.be}
\emailAdd{seminara@fi.infn.it}
\emailAdd{i.yaakov@soton.ac.uk}

\abstract{
  We compute the exact all-orders perturbative expansion for the partition function of 2d $\mathrm{SU}(2)$ Yang--Mills theory on closed surfaces around higher critical points. We demonstrate that the expansion can be derived from the lattice partition function for all genera using a distributional generalization of the Poisson summation formula. We then recompute the expansion directly, using a stationary phase version of supersymmetric localization. The result of localization is a novel effective action which is itself a distribution rather than a function of the supersymmetric moduli. We comment on possible applications to A-twisted models and their analogs in higher dimensions. 
}

\usepackage{amsthm}
\newtheorem*{proposition}{Proposition}

\begin{document} 
\maketitle
\flushbottom

\section{Introduction}

Pure Yang--Mills theory in two spacetime dimensions (2dYM) is one of only a handful of examples of exactly solvable interacting gauge theories, another being the closely related Chern--Simons theory in three dimensions.
Solvability can be traced back to any number of properties: the lack of propagating degrees of freedom; the existence of a subdivision invariant lattice model \cite{Migdal:1975zg}; invariance under area preserving diffeomorphisms \cite{Witten:1991we}; a type of hidden supersymmetry \cite{Witten:1992xu}; the existence of a gauge which allows one to explicitly evaluate the path integral \cite{Blau:1993hj}; and possibly others.

Despite this simplicity, 2dYM resembles its higher dimensional cousins in some respects. In the large-$N$ limit, for instance, the theory exhibits behaviors that are expected to hold true also in higher dimensions: it can be reformulated exactly as a string theory \cite{Gross:1992tu, Gross:1993hu}, and experiences phase-transitions at critical values of the coupling \cite{Douglas:1993iia, Gross:1994mr}.

At finite gauge coupling, 2dYM displays a number of interesting mathematical phenomena whose detailed study may be relevant for more physical models. These phenomena include an asymptotic, but not generally convergent, formal perturbation series. Additionally, a plethora of results connecting 2dYM on the two-sphere, limited to the zero instanton sector, with perturbative computations in $2+\epsilon$ dimensions \cite{Bassetto:1994vj, Bassetto:1998sr}, and with exact results for $\mathcal{N}=4$ Super Yang--Mills in four dimensions \cite{Drukker:2007qr,Pestun:2009nn,Giombi:2009ds,Bassetto:2009rt,Giombi:2018qox}. We also have detailed knowledge of non-perturbative phenomena: generically unstable instantons---non-trivial solutions of the Yang-Mills equation---contribute to both the partition function \cite{Witten:1992xu}, and the expectation value of other observables \cite{Gross:1994ub, Giombi:2009ek}.

The limit where the gauge coupling vanishes is also interesting, because 2dYM becomes a topological quantum field theory (TQFT).
In \cite{Witten:1991we}, Witten related the partition function of topological 2dYM on a closed Euclidean manifold to the symplectic volume of the moduli space of flat connections for the gauge group.
The followup \cite{Witten:1992xu} further demonstrated that the intersection numbers of this moduli space can be conveniently computed by recasting topological 2dYM as a twisted supersymmetric gauge theory.
In this supersymmetric setup, a number of observables in 2dYM turn out to be protected, that is to preserve some or all of the supersymmetry. These include the partition function on closed Euclidean manifolds, and some operators, both local and extended.
Curiously, the Wilson loop is not included in the list of protected operators, despite being fully calculable in other approaches \cite{Rusakov:1990rs}.

The computation in the topological theory relies on a relationship between observables in the topological theory and those in a supersymmetric version of 2dYM at finite coupling (the ``physical theory'').
Unstable configurations, which contribute to observables in this theory, do not usually make contributions to protected observables.
In the present context, their appearance in the localization computation was explained in \cite{Witten:1992xu}, in the process of introducing non-abelian localization.

The existence of a non-trivial perturbation series could also seem in tension with generic expectations for protected observables. The vacuum expectation value of a protected observable would be expected to depend on the coupling constant only to some finite order, usually to one loop. Loops around non-perturbative BPS instanton contributions, which can be shown to arise in supersymmetric localization \cite{Pestun:2007rz}, are similarly limited. However, as is the case with 2dYM, the field theoretic coupling constant does not always coincide with the deformation parameter used in localization, leading to interesting results.

The nature of the weak-coupling expansions in supersymmetric theories have been studied in a number of papers, mainly in connection with the resurgence program \cite{Aniceto:2018bis}.
They exhibit a variety of behaviors that depend on the theory and the space-time dimension, including truncation at a finite order \cite{Dorigoni:2017smz,Dorigoni:2019kux,Fujimori:2022qij}; a finite radius of convergence \cite{Aniceto:2014hoa}; and asymptoticity associated with Borel \cite{Aniceto:2014hoa,Honda:2016mvg,Honda:2017cnz} and non-Borel summability \cite{Honda:2016vmv}.

In all the known cases, supersymmetric localization yields a non-trivial perturbative expansion by expanding the one-loop determinants around small values of the moduli. In this sense the computation is still one-loop exact, being determined by the semiclassical approximation induced by the localizing term. The perturbative series appearing in 2dYM seems, in this respect, somewhat different: the perturbative expansion being induced by a \emph{dual} expansion in the moduli.

In this paper, we will investigate the appearance of a non-trivial perturbation series in 2dYM. Our explanation will come from a combination of two different strategies for extracting the series, whose results match up to some plausible assumptions regarding localization.

First, following \cite{Witten:1992xu}, we will use Poisson summation to convert the well-known form of the partition function, in terms of irreducible representations, into a form where the perturbation series around each instanton is manifest. For genus greater than $1$, this requires the use of a recent generalization of the ordinary Poisson summation formula, which applies to distributions rather than functions. We will then reverse engineer the result in order to extract the elements of a localization computation. Specifically, we will extract an effective action---or one-loop determinant---for the constant mode of the scalar field in the \emph{basic multiplet} described by Witten in \cite{Witten:1992xu}.

The second strategy will be to perform a somewhat non-standard localization computation of 2dYM, involving oscillatory rather than exponentially suppressed path integrals. Among the moduli arising from this computation will be the same scalar field mode referenced above. We will argue that the asymptotic expansion of these integrals, which is all that contributes to the localization result, yields an effective action for this mode in the form of a tempered distribution. Integration over the moduli will then lead to a result compatible with Poisson summation.

We will briefly review Witten's results for the partition function, using both the lattice formulation \cite{Witten:1991we} and non-abelian localization \cite{Witten:1992xu}, in Section~\ref{SEC:2dYM}. In Section~\ref{SEC:Z_from_poisson_summation}, by means of a generalization of the conventional Poisson summation formula, we will extract the full perturbative expansion associated with each instanton sector for the $\mathrm{SU}(2)$ partition function. In Section~\ref{SEC:localization}, we will argue that the same result can be obtained through a localization computation in which the contribution of an arbitrary instanton sector is associated with the integral over the fluctuations around higher critical points of the localizing action. At the beginning of this section, we review Witten's argument and apply it to a toy model in finite dimensions. Section~\ref{SEC:discussion} contains our conclusions and possible further developments of our investigations. Some technical aspects of the paper are contained in two Appendices.

\section{Two-dimensional Yang--Mills theory}\label{SEC:2dYM}

In this section, we provide a brief introduction to the relevant aspects of two-dimensional Yang--Mills theory (2dYM), following \cite{Witten:1991we}.
The section includes no new results.
We will also use facts about the space of connections on a principle $G$-bundle over a Riemann surface $\Sigma$, which were derived by Atiyah and Bott in \cite{Atiyah:1984px,Atiyah:1982fa}. Additional results specific to Yang--Mills connections are given in Appendix~\ref{sec:Yang-Mills-connections}.

\paragraph{The space of connections.}

Let $G$ be a compact simple connected Lie group.
We denote its Lie algebra by $\mathfrak{g}$, and follow the mathematics conventions whereby $\mathfrak{g}$ is represented by \emph{anti-Hermitian} matrices.
Moreover, $\mathfrak{g}$ admits a $G$-invariant inner product $\langle\cdot,\cdot\rangle$ which is unique up to rescaling.
For $\mathrm{SU}(N)$, this inner product coincides with the trace in the fundamental representation:
\begin{align}
  \langle a,b\rangle \equiv -\tr(ab) \;.
\end{align}

Let $\Sigma$ be an oriented Riemannian manifold of dimension $2$, and let $A$ be a connection on a principle $G$-bundle over $\Sigma$. $A$ can be represented locally by and adjoint-valued one-form. The curvature of $A$, given by
\begin{align}
  F = \d A + A\wedge A \;,
\end{align}
is an adjoint-valued two-form.
The set $\mathcal{A}$ of all connections on a fixed principle bundle over $\Sigma$ is an infinite-dimensional affine space. Let $\delta A$ denote a tangent vector at some point in this space. $\delta A$ is a globally-defined adjoint-valued one-form on $\Sigma$. The Atiyah--Bott symplectic form on $\mathcal{A}$ is defined as follows \cite{Atiyah:1984px,Atiyah:1982fa}:
\begin{align}
  \omega_{\text{AB}}(\delta A,\delta B) \equiv \frac{1}{4\pi^{2}} \int_{\Sigma} \tr(\delta A\wedge\delta B) \;.
\end{align}
Note that $\omega_{\text{AB}}$ is constant, i.e.\ independent of the connection, and thus formally closed. 

Using the metric on $\Sigma$, and the associated Hodge star operator $\star$, one can also define on $\mathcal{A}$ an integrable complex structure $J$ and an inner product $(\cdot,\cdot)$ as
\begin{align}
  J \, \delta A
    &\equiv -\star\delta A \;, \\
  (\delta A,\delta B)
    &\equiv -\frac{1}{4\pi^{2}}\int_{\Sigma}\tr(\delta A\wedge\star\,\delta B) \;.
\end{align}
$\omega_{\text{AB}},J$ and the metric associated to $(\cdot,\cdot)$ form a compatible triple
\begin{align}
  (\delta A,\delta B) = \omega_{\text{AB}}(\delta A,J\,\delta B)\,,
\end{align}
and formally endow $\mathcal{A}$ with a K\"ahler structure. 

The group of gauge transformations $\mathcal{G}$, automorphisms of a fixed principle bundle, acts symplectically on the space of connections, with a moment map given by
\begin{align}
  \mu(A) = -\frac{F}{4\pi^{2}}\,.
\end{align}
The moduli space of flat connections is formally given by the infinite dimensional symplectic reduction of the space of connections with vanishing curvature by the symplectic action of the gauge group
\begin{align}
  \mathcal{M}_{\text{flat}} \equiv \mathcal{A}/\!\!/\mathcal{G} = \{A\,|\,\mu(A)=0\} / \mathcal{G} \;.
\end{align}

\paragraph{The theory.}

The action of 2dYM reads
\begin{equation}\label{eq:Yang_Mills_action}
  S_{\text{YM}}[A] = -\frac{1}{8\g}\int_{\Sigma}\tr(F\wedge\star F) \;,
\end{equation}
where $\g$ is a coupling constant with dimensions of a mass squared,\footnote{
  Our coupling constant $\g$ is related to the one used by Witten in \cite{Witten:1992xu} as $\g=\pi^2\epsilon_\text{Witten}$.}
and $\tr$ is the trace in the fundamental representation of $\mathrm{SU}(N)$.
The equations of motion stemming from \eqref{eq:Yang_Mills_action} are the Yang--Mills equations 
\begin{align}
  \d_A\star F = 0 \;,
\end{align}
where $\d_A = \d + A$ is the gauge-covariant extension of the exterior derivative. 

The partition function of quantum 2dYM is defined as 
\begin{align}
  Z \equiv \frac{1}{|Z(G)|}\int\mathcal{D}A \; e^{-S_{\text{YM}}[A]} \;,
\end{align}
in terms of an appropriately gauge-fixed version of the path integral.
$|Z(G)|$ is the volume of the center of $G$. We will consider only examples in which $Z(G)$ is finite, in which case $|Z(G)|$ is the order of $Z(G)$.
The action $S_{\text{YM}}$ is super-renormalizable and does not require field-dependent counterterms.
$Z$ is defined up to scheme-dependent finite counterterms with parameters $u_{0}, v_{0}$, which multiply $Z$ by a factor
\begin{align}\label{eq:allowed_renormalization}
  e^{u_{0}\g\rho(\Sigma) + v_{0}\chi(\Sigma)} \;,
\end{align}
where $\chi(\Sigma) \equiv 2-2g$ is the Euler characteristics of $\Sigma$, and $\rho(\Sigma)$ is the total area.

The theory can be re-written in an equivalent form by introducing an adjoint-valued scalar $\phi$, and using the action 
\begin{align}\label{eq:YM_BF_action}
  S_{\text{BF}}[A,\phi] = -\frac{\i}{4\pi^2}\int_{\Sigma}\tr(\phi F) - \frac{\g}{8\pi^4}\int_{\Sigma}\vf\,\tr(\phi^{2}) \;,
\end{align}
where with $\vf$ we denote the metric volume form on $\Sigma$.

This formulation makes manifest an important feature of the theory: the action \eqref{eq:YM_BF_action} depends on the geometry of $\Sigma$ only through $\vf$. As a consequence, the resulting theory is only sensitive to the area
\begin{align}
  \rho(\Sigma) \equiv \int_{\Sigma} \vf \;,
\end{align}
the unique invariant of such a measure in dimension $2$.
Moreover, the Yang--Mills coupling and the volume form appear only through the combination $\g\,\vf$. This implies that the theory depends on the coupling only through the dimensionless combination $\g\,\rho(\Sigma)$.

The limit $\g\to0$, or equivalently $\rho(\Sigma)\to0$, defines a BF type TQFT. Witten showed in \cite{Witten:1991we} that the partition function in this limit is proportional to the symplectic volume of the moduli space $\mathcal{M}_{\text{flat}}$ of flat $G$ connections on $\Sigma$, relative to the Atiyah--Bott symplectic form \cite{Atiyah:1982fa}.\footnote{
  The proportionality constant consists of the number of elements in the center of $G$, a normalization factor coming from the relationship between the symplectic form and the path integral measure, and a standard renormalization of the form $e^{v_{0}\,\chi(\Sigma)}$ (cf.\ Eq 4.72 of \cite{Witten:1991we}).
}

\paragraph{The lattice gauge theory approach.}

2dYM admits a discretized lattice gauge theory model with the standard degrees of freedom: an element $U\in G$ living on every edge \cite{Wilson:1974sk}.
The lattice Yang--Mills action in two dimensions takes a particularly simple form \cite{Migdal:1975zg}.
Let $\mathfrak{R}$ denote the set of irreducible representations of $G$, $\dim\alpha$ the dimension of $\alpha\in\mathfrak{R}$, and $c_{2}(\alpha)$ the value of the second Casimir.
Let $\mathcal{U}_P$ be the holonomy around an elementary lattice plaquette $P$, i.e.\ the ordered product of the $U$'s, and $\chi_{\alpha}(\mathcal{U}_P)$ its character in the representation $\alpha$.
The lattice Yang--Mills action can be written in this basis as
\begin{equation}\label{eq:Migdal_measure}
  \Gamma = \sum_{\alpha\in\mathfrak{R}} \dim\alpha \; e^{-\g\,\rho(P)\,c_{2}(\alpha)} \; \chi_{\alpha}(\mathcal{U}_P) \;,
\end{equation}
and the lattice partition function can be written as 
\begin{align}
  Z_{\text{lattice}} = \int \prod_{e\in\{\text{edges}\}} \d U_{e} \; \prod_{P\in\{\text{plaquettes}\} }\Gamma(\mathcal{U}_P,\rho(P)) \;.
\end{align}
With the measure \eqref{eq:Migdal_measure}, the lattice gauge theory is invariant under subdivisions of the lattice.
The partition function on a discretized surface can therefore be computed using a single plaquette, and the continuum limit is trivial \cite{Migdal:1975zg,Witten:1991we}.
The partition function on a genus $g$ surface can be shown using this approach to be \cite{Witten:1991we}
\begin{align}\label{eq:Migdal_partition_function}
  Z_{\text{lattice}} = \sum_{\alpha\in\mathfrak{R}} (\dim\alpha)^{2-2g} \; e^{-\g\,\rho(\Sigma)\,c_{2}(\alpha)} \;.
\end{align}
For convenience and without loss of generality, in the rest of this paper we will adimensionalize the Yang--Mills coupling $\epsilon$ by working in units where $\rho(\Sigma) = 1$.

\paragraph{The supersymmetric theory and localization.}
2dYM admits an equivalent formulation with a fermionic symmetry. Let $\psi$ be a Grassmann-valued and adjoint-valued one-form field. The action 
\begin{align}\label{eq:cohomological_Yang_Mills_action}
  S_{\text{SUSY}}[A,\psi,\phi] = -\frac{1}{4\pi^2}\int\tr\Big(\i\phi F + \frac{1}{2}\,\psi\wedge\psi\Big) -\frac{\g}{8\pi^4}\int \vf\,\tr(\phi^{2})
\end{align}
yields a theory which is equivalent to 2dYM after integrating out $\phi$ and $\psi$, and hence yields the same partition function.
$S_{\text{SUSY}}$ is also invariant under the following odd symmetry transformation:
\begin{align}\label{EQ:SUSY_transformations}
  \delta A &= \i\psi \;, &
  \delta\psi &= -\d_{A}\phi \;, &
  \delta\phi &= 0 \;.
\end{align}
The transformation $\delta$ squares to an infinitesimal gauge transformation with parameter $\phi$
\begin{align}
  \delta^{2} = -\i\,\delta_{\phi} \;.
\end{align}

In \cite{Witten:1992xu}, it was shown that the action \eqref{eq:cohomological_Yang_Mills_action} admits a $\delta$-exact deformation which can be used to localize the theory. The specific deformation used by Witten, which requires the introduction of two additional supermultiplets, was shown to localize the theory to the moduli space of solutions to the Yang--Mills equations on $\Sigma$. The connected components of this moduli space, which deserve the name instanton sectors, are labeled by $\Gamma$, the magnetic weight lattice of $G$, which is dual to $\Gamma^{*}$, the usual weight lattice of $G$, modulo the action of the Weyl group $\mathcal{W}$ \cite{Goddard:1976qe,Atiyah:1982fa}.\footnote{
  For a nice introduction, see the appendix of \cite{Kapustin:2005py}.}
The localization procedure used by Witten, which differs from standard localization computations, was given the name ``non-abelian localization''.\footnote{
  The difference is visible, for instance, in the fact that standard localization computations yield a moduli space which consists of solutions to \emph{first order} BPS equations.}

For $g\ge2$, the value of the partition function derived using non-abelian localization is, schematically,
\begin{equation}\label{eq:Witten_partition_function}
  Z = \sum_{\m\in\Gamma/\mathcal{W}} q_g(\m,\g) \; e^{-S(\m)} \;,
\end{equation}
where
\begin{align}
  S(\m) \propto \frac{\langle\m,\m\rangle}{\g}
\end{align}
is the value of the action \eqref{eq:Yang_Mills_action} on any solution of the Yang--Mills equations in the component labeled by $\m$, and $q_g(\m,\g)$ is a function which was not explicitly computed in \cite{Witten:1992xu}.
When $\mathcal{M}_{\text{flat}}$ is smooth, the part of \eqref{eq:Witten_partition_function} which is not exponentially suppressed in the limit $\g\rightarrow0$ is in fact a polynomial in $\g$, which we call $p_g(\g)$. Witten argued that $p_g(\g)$ contains all the information required in order to compute the intersection theory on $\mathcal{M}_{\text{flat}}$. 

Recall that the set of irreducible representations of $G$, denoted by $\mathfrak{R}$, is equivalent to the weight lattice of $G$, which we denoted by $\Gamma^*$, modulo the action of the Weyl group.
For $G=\mathrm{SU}(2)$, $\Gamma$ and $\Gamma^*$ are one dimensional, and the value of the partition function as given in \eqref{eq:Migdal_partition_function} can be written simply as\footnote{
  Precisely, \eqref{eq:Migdal_partition_function_su2} is obtained from \eqref{eq:Migdal_partition_function} by adding a customary constant term to the Hamiltonian (cf.\ Eq.~(4.39) in \cite{Witten:1992xu}).
}
\begin{equation}\label{eq:Migdal_partition_function_su2}
  Z = \sum_{\ell=1}^{\infty} \ell^{2-2g} \; e^{-\g\ell^2} \;,
\end{equation}
while the version in \eqref{eq:Witten_partition_function} can be written as
\begin{equation}\label{eq:Witten_partition_function_su2}
  Z = \sum_{\m=0}^{\infty} q_g(\m,\g) \; e^{-\pi^{2}\m^{2}/\g} \;.
\end{equation}

Witten argued in \cite{Witten:1992xu} that the expressions \eqref{eq:Migdal_partition_function_su2} and \eqref{eq:Witten_partition_function_su2} are equal for all values of $\g$ and $g$. However, the precise matching is only studied in the sector where $\m$ vanishes.
We will, instead, investigate the functions $q_g(\m,\g)$ for all values of $\m$ and demonstrate explicitly the equality of these expressions up to an overall factor. Furthermore, we will show that these functions admit asymptotic expansions in $\g$, which can be interpreted as determining the formal perturbation series around each instanton.

\section{\boldmath The \texorpdfstring{$\mathrm{SU}(2)$}{SU(2)} partition function}\label{SEC:Z_from_poisson_summation}
\subsection{From Poisson summation}
In this section, we will use a generalized version of the Poisson summation formula to recast the $\mathrm{SU}(2)$ partition function \eqref{eq:Migdal_partition_function_su2}, expressed as a sum over representations of the gauge group, into a sum over the critical points of the Yang--Mills action, as in \eqref{eq:Witten_partition_function_su2}.\footnote{
  In \cite{Fujimori:2022qij}, the sum over instantons is studied through resurgence. However, the results therein, although similar in form, differ from the ones presented here and fail to fully reproduce the vacuum sector computed in \cite{Witten:1992xu, Ghoshal:2014kpa}. 
}
For non-negative Euler characteristics we can regard the argument of the sum in \eqref{eq:Migdal_partition_function_su2} as the even function
\begin{align}\label{EQ:f_g<=1}
  f(x) =x^{2-2g} \; e^{-{\g} x^2}
\end{align}
belonging to the Schwarz space $\mathcal{S}(\mathbb{R})$. Now, denoting with
\begin{align}
  \hat{f}(p) = \int_{-\infty}^{+\infty}\d x \; e^{-2\pi\i px} \, f(x)
\end{align}
its Fourier transform, one can employ the standard Poisson summation to write\footnote{In the last step of \eqref{EQ:poisson} we used the fact that the Fourier transform of an even function is itself an even function.}
\begin{align}\label{EQ:poisson}
  \sum_{\ell\geq1} f(\ell)
  &= - \frac{1}{2}\, f(0) + \frac{1}{2} \sum_{x\in\mathbb{Z}} f(x) \cr
  &= - \frac{1}{2}\, f(0) + \frac{1}{2} \sum_{p\in\mathbb{Z}} \hat{f}(p) \cr
  &= - \frac{1}{2}\, f(0) + \frac{1}{2}\, \hat{f}(0) + \sum_{\m\geq1} \hat{f}(\m) \;.
\end{align}
This gives the genus-zero and genus-one partition functions
\begin{align}
  Z\big|_{g=0} &= \label{EQ:Z_0}
    \frac{\sqrt{\pi}}{4}\g^{-3/2} + \sum_{\m=1}^{\infty} e^{-\pi^2\m^2/\g} \, \bigg(\frac{\sqrt{\pi}}{2}\,\g^{-3/2} - \pi^{5/2}\g^{-5/2}\m^2\bigg) \;, \\
  Z\big|_{g=1} &= \label{EQ:Z_1}
    \frac{\sqrt{\pi}}{2}\g^{-1/2} - \frac{1}{2} + \sum_{\m=1}^{\infty} e^{-\pi^2\m^2/\g} \, \bigg(\sqrt{\pi} \, \g^{-1/2}\bigg) \;.
\end{align}

As mentioned in the previous section, from \eqref{EQ:Z_0} and \eqref{EQ:Z_1} one can observe exponentially suppressed terms controlled by the classical action evaluated on the critical points of the action \eqref{eq:Yang_Mills_action}. In fact, for the gauge group $\mathrm{SU}(2)$, a covariantly constant $F$ reduces the structure group to $\mathrm{U}(1)$, splitting the $\mathrm{SU}(2)$ bundle into the sum of two line bundles. A representative of the conjugacy class of $F$ is given by
\begin{align}\label{EQ:critical_point}
  \star F = 2\pi\i\m\,\sigma_3 \;,
\end{align}
where $\sigma_3$ denotes the third Pauli matrix.
Configurations with opposite signs of $\m$ are gauge equivalent, as they are connected by a Weyl reflection. In labelling Yang--Mills solutions modulo gauge transformation, we restrict $\m$ to non-negative values.
The value of the action on the critical point \eqref{EQ:critical_point} is exactly $S=\pi^2\m^2/\g$. Moreover, from \eqref{EQ:Z_0} and \eqref{EQ:Z_1} one can directly read off the $q_0(\m,\g)$ and $q_1(\m,\g)$ polynomials in \eqref{eq:Witten_partition_function_su2}.

When $g\geq2$, the standard Poisson summation formula fails due to the fact that $f(x)$ diverges at $x=0$, and its Fourier transform does not exist in the usual sense. To proceed, we have to resort to a distributional extension of the Poisson summation formula proven in \cite{DURAN1998581}, where the singular points are excluded from the sum. Since $f(x)$ has a power-like divergence at $x=0$, and it is of rapid decay, it can be viewed as an element of the space $\mathcal{S}'(\mathbb{R})$ of tempered distribution. According to \cite{DURAN1998581}, we can rewrite the identity \eqref{EQ:poisson} as 
\begin{align}\label{EQ:generalized_Poisson}
  \sum_{x\in\mathbb{Z}\setminus\{0\}} f(x) = \Clim_{\Lambda\to\infty}\Bigg(\sum_{|p|\leq\Lambda} \hat{f}(p) - \int_{\raisebox{-4pt}{\scriptsize$|p|{\leq}\Lambda$}} \!\! \d p \; \hat{f}(p) \Bigg) \;,
\end{align}
where the limit is taken in the Ces\`aro sense.\footnote{
  Let $a\in\mathbb{R}\setminus\{-1,-2,-3,\dots\}$. We say that $f(x) = o(x^a)$ in the Cesàro sense for $x\to\infty$ if there exists $n\in\mathbb{N}$, a primitive $F(x)$ of order $n$ of $f(x)$, and a polynomial $p(x)$ of degree $n-1$ such that $F(x)$ is locally integrable for large $x$ and $$F(x) = p(x) + O(x^{n+a})$$ as $x\to\infty$. Then, we say that $$\Clim_{x\to\infty} f(x) = L$$ if $f(x) - L = o(1)$ in the Cesàro sense for $x\to\infty$. Obviously, if something converges or has a definite limit in the usual sense, it also converges or possesses a well-definite limit in the Cesàro sense, and the two coincide.
  
  For a more detailed definition and discussion, we refer to \cite{DURAN1998581} and references therein.
}

In our case, $f(x)$ can be regarded as the tempered distribution\footnote{
  To be more precise, in Theorem~1 of \cite{DURAN1998581}, \eqref{EQ:generalized_Poisson} is shown to hold in a subset of $\mathcal{S}'(\mathbb{R})$ to which \eqref{EQ:f} belongs. This is the space $\mathcal{K}+\hat{\mathcal{K}}$, where $\mathcal{K}$ is the space of smooth function for which there exist a $q\in\mathbb{R}$ such that
  $$\frac{\partial^k}{\partial x^k} f(x)=O(|x|^{q-k})$$
  as $|x|\to\infty$, while $\hat{\mathcal{K}}$ is the space of their Fourier transforms, which exist in the Ces\`aro sense. The elements of $\hat{\mathcal{K}}$ are smooth functions in $\mathbb{R}\setminus\{0\}$ rapidly decaying at infinity. On the other hand, because of the singular behavior in the origin, a generic $f \in \mathcal{K}+\hat{\mathcal{K}}$ is a tempered distribution rather than a function.

  For a version of the formula that applies to any distribution in $\mathcal{S}'(\mathbb{R})$, see Theorem~2 of \cite{DURAN1998581}.
}
\begin{align}\label{EQ:f}
  f(x)
  &= -\frac{e^{-\g x^2}}{(2 g-3)!} \, \frac{\d^{2g-3}}{\d x^{2 g-3}} \left(\pv \frac{1}{x}\right) \;, \qquad \text{where $g\geq2$} \;.
\end{align}
Its Fourier transform has been computed in the distributional sense in Appendix~\ref{APP:Fourier} and can be written as the sum of two contributions,
\begin{align}\label{EQ:fhat_pre}
  \hat{f}(p) = Q'(p) + H(p) \;,
\end{align}
where\footnote{
  The \emph{complementary error function} is defined as $$\erfc(x) = 1 - \erf(x) \;.$$}
\begin{align}\label{EQ:H}
  H(p) = e^{-\pi^2p^2/\g} \, P(p) - \erfc(\pi p/\sqrt{\g}) \, Q'(p) \;.
\end{align}
The two coefficients $P(p)$ and $Q(p)$ are even polynomials whose definition can be found in \eqref{EQ:P} and \eqref{EQ:Q}. Since $\hat{f}$ is an even function, we can also write
\begin{align}\label{EQ:fhat}
  \hat{f}(p)
  &= Q'(|p|) + H(|p|) \;.
\end{align}
The benefit of having split $\hat{f}$ as in \eqref{EQ:fhat} is that we can apply the limiting procedure to the two terms separately. Since $H(|p|)$ is exponentially suppressed\footnote{
  We have, in particular, that $\erfc(z)\sim e^{-z^2}$ for $z\to\infty$.}
for $p\to\pm\infty$, the associated sum and integral converge separately as $\Lambda\to\infty$. In other words, the only nontrivial Ces\`aro limits originate from the $Q'(|p|)$ term. We can then rewrite the partition function as the sum
\begin{align}\label{EQ:Zg_from_RSI}
  Z
  &= \frac{1}{2}\left(R + I + S\right)
\end{align}
of the three finite terms
\begin{align}
  R \label{EQ:R}
  &= \Clim_{\Lambda\to\infty}\Bigg(\sum_{|p|\leq\Lambda} Q'(|p|) - \int_{\raisebox{-4pt}{\scriptsize$|p|{\leq}\Lambda$}} \!\! \d p \; Q'(|p|) \Bigg) \;, \\
  I \label{EQ:I}
  &= - \int_{-\infty}^{+\infty} \d p \; H(|p|) \;,
  \vphantom{\sum_{|p|\leq\Lambda}} \\
  S \label{EQ:S}
  &= \sum_{p\in\mathbb{Z}} H(|p|)
  \vphantom{\int^{\infty}} \;.
\end{align}
Let us now deal with them one by one.

We start by considering $R$. One can avoid performing computations at finite $\Lambda$ by simply applying again the extended Poisson summation, but this time in the opposite direction. Exploiting the Fourier transform computed in Appendix~\ref{APP:Fourier}, we get
\begin{align}\label{EQ:generalized_Poisson_monomial}
  \Clim_{\Lambda\to\infty}\Bigg(\sum_{|p|\leq\Lambda} |p|^{2n+1} - \int_{\raisebox{-4pt}{\scriptsize$|p|{\leq}\Lambda$}} \!\! \d p \; |p|^{2n+1} \Bigg)
  &= \frac{(2n+1)!}{2^{2n+1}\,(-\pi^2)^{n+1}} \sum_{\ell\in\mathbb{Z}\setminus\{0\}} \frac{1}{\ell^{2n+2}} \cr
  &= \frac{(2n+1)!}{2^{2n}\,(-\pi^2)^{n+1}} \; \zeta(2n+2) \;.
\end{align}
By linearity, we can simply apply the above to each monomial in $Q'$ to find
\begin{align}\label{EQ:R_result}
  R
  &= 2 \sum_{j=0}^{g-2} \frac{(-1)^j}{j!} \, \zeta(2g-2-2j) \; \g^j \;.
\end{align}
Since the coefficients depend only on  the Riemann zeta function evaluated on even integers, we can alternatively express them in terms of Bernoulli numbers (see e.g.\ \cite{Witten:1992xu}).

The next step is to compute $I$. Plugging \eqref{EQ:H} into \eqref{EQ:I} and integrating the second term in \eqref{EQ:H} by parts, we find
\begin{align}\label{EQ:R_computation}
  I
  &= 2\bigg[\erfc(\pi p/\sqrt{\g}) \, Q(p)\,\bigg]_0^\infty - 2\int_0^\infty \d p \; e^{-\pi^2p^2/\g} \bigg(P(p) - \frac{2\sqrt{\pi}}{\sqrt{\g}} \, Q(p) \bigg) \;.
\end{align}
The residual integral in \eqref{EQ:R_computation} can be directly computed and shown to vanish identically. The boundary contribution, instead, yields 
\begin{align}
  I
  &= -2 \, Q(0) \cr
  &= -\frac{(-1)^{g-1}\,\g^{g-1}}{(g-1)!} \;.
\end{align}
This contribution amounts to the term with $j=g-1$ missing from the summation range in $R$, as given in \eqref{EQ:R_result}.

Lastly, we need to compute $S$. Let us rewrite \eqref{EQ:S} as
\begin{align}
  S &= P(0) + 2\sum_{\m=1}^{\infty} H(\m) \;.
\end{align}
One can easily compute
\begin{align}
  P(0)
  = \frac{(-1)^{g-1}\sqrt{\pi}\,2^{g-1}}{(2g-3)!!}\,\g^{g-3/2}
\end{align}
by taking only the $k=g-2$ term in \eqref{EQ:P} and performing the sum over $m$.
To rewrite the  remaining term as a sum over the critical points with a non-vanishing magnetic flux, we use the asymptotic expansion of the complementary error function
\begin{align}
  \erfc(1/z)
  &= e^{-1/z^2} \frac{z}{\sqrt{\pi}} \sum_{m=0}^\infty (-z^2/2)^{m}\,(2m-1)!! \;,
\end{align}
valid for $|\arg(z)|<\frac{3}{4}$. This gives
\begin{align}\label{EQ:H_asymptotic_expansion}
  H(|p|)
  &= e^{-\pi^2p^2/\g} \sum_{k=0}^{\infty} \frac{\sqrt{\pi}\,(2g-2)_{2k}}{2^{2k}\,(-\pi^2p^2)^{g-1+k}\,k!} \; \g^{2g-5/2+k} \;.
\end{align}

We can now add up all the different contributions as in \eqref{EQ:Zg_from_RSI} to obtain the complete partition function of $\mathrm{SU}(2)$ Yang--Mills theory at genus $g\geq 2$:
\begin{align}\label{EQ:Zg_g≥2}
  Z
  = {}&\sum_{j=0}^{g-1} \frac{(-1)^j}{j!} \, \zeta(2g-2-2j) \; \g^{j} + \frac{(-1)^{g-1}\sqrt{\pi}\,2^{g-2}}{(2g-3)!!}\,\g^{g-3/2} \cr
  &+ \sum_{\m=1}^{\infty} e^{-\pi^2\m^2/\g} \, \sum_{k=0}^{\infty} \frac{\sqrt{\pi}\,(2g-2)_{2k}}{2^{2k}\,(-\pi^2\m^2)^{g-1+k}\,k!} \; \g^{2g-5/2+k}  \;.
\end{align}
The terms in the first line stem from flat connections wrapping the $2g$ non-trivial cycles of $\Sigma$. The polynomial part evaluates the intersection numbers of the moduli space of flat connections (once the partition function has been properly normalized), while the non-analytic part with its exponent $g-\frac{3}{2}$ originates from the singular nature of the moduli space of flat $\mathrm{SU}(2)$ connections. This, in turn, is a consequence of the that fact that $\mathrm{SU}(2)$ has a nontrivial center and thus the action of the $\mathrm{SU}(2)$ gauge group on the space of flat connections is not free.
The terms in the second line of \eqref{EQ:Zg_g≥2} are governed once again by the exponential of the classical action evaluated at the critical points. At variance with the $g\leq1$ case, the fluctuations around the instanton solutions become (infinite) asymptotic series rather than polynomials in $1/\sqrt{\g}$.

\subsection{From an effective action}
In \eqref{EQ:Zg_g≥2}, we have isolated the contribution to the partition function associated with each instanton sector. Later in this paper, we will deal with the problem of deriving the same result from a localization computation.
Before closing this section, however, we want to illustrate how the analysis performed so far already contains some insight on the form of the effective action that should arise from the localization procedure.

As mentioned, we can recognize in $H(|\m|)$ the contribution of a single instanton with flux number $\m$.
From the perspective of a localization computation, we expect this to arise from the integral over a bosonic moduli (i.e.\ the \textit{localizing saddle}) corresponding to some supersymmetric field configurations, the theory being formulated in terms of an action with a fermionic symmetry. The minimal supersymmetric extension for 2dYM comprises the degrees of freedom that appear in the classical action \eqref{eq:cohomological_Yang_Mills_action}. The associated supersymmetric variations \eqref{EQ:SUSY_transformations} indicate that a putative effective action should depend, not just on $\m$, but also on a constant mode for the scalar $\phi$ in the Cartan subalgebra of $\mathrm{SU}(2)$. For simplicity, for the rest of this section, we will denote this single real degree of freedom with $\phi_0$, where $\phi = 2\pi^2\i\phi_0\,\sigma_3$.

From the discussion above, we expect the function $H$ to be generated by an integral of the type
\begin{align}\label{EQ:H_from_localization}
  H(|\m|) = \int_{-\infty}^{+\infty} \d\phi_0 \; e^{-S_{\text{cl.}}(\phi_0,\m)} \; \Delta^{\vphantom{\text{r}}}_g(\phi_0,\m) \;,
\end{align}
where
\begin{align}\label{EQ:S_cl}
  S_{\text{cl.}}(\phi_0,\m) = \g\phi_0{}^2 + 2\pi\i\m\phi_0
\end{align}
is the bosonic part of the supersymmetric action \eqref{eq:cohomological_Yang_Mills_action} evaluated on the bosonic moduli, while $\Delta^{\vphantom{\text{r}}}_g$ is the one-loop determinant obtained by integrating quantum fluctuations about each point in the moduli.

Interestingly, the results of this section hint at a distributional nature for $\Delta_g$. In fact, the integral \eqref{EQ:H_from_localization} closely resembles the Fourier transform computed in Appendix~\ref{APP:Fourier}. Specifically, \eqref{EQ:fhat_pre} suggests that we should write
\begin{align}
  \Delta^{\vphantom{\text{r}}}_g(\phi_0,\m) = e^{\g\phi_0{}^2} f(\phi_0) + \Delta^{\text{res}}_g(\phi_0,\m) \;,
\end{align}
where
\begin{align}\label{EQ:Delta_res}
  \int_{-\infty}^{+\infty} \d\phi_0 \; e^{-2\pi\i\m\phi_0-\g\phi_0{}^2} \, \Delta^{\text{res}}_g(\phi_0,\m) = - Q'(|\m|) \;.
\end{align}
From the representation \eqref{EQ:Q'} of $Q'$ in terms of a Hermite polynomial, we obtain\footnote{
  In the first line of \eqref{EQ:Q'_as_Fourier_trasnform}, we have made use of the expression for the generating function of the Hermite polynomials 
  $$e^{2xt-t^2} = \sum_{n=0}^\infty H_n(x) \, \frac{t^n}{n!} \;.$$
}
\begin{align}\label{EQ:Q'_as_Fourier_trasnform}
  -Q'(\m)
  &= \frac{\i\pi}{(2g-3)!} \, \frac{\d^{2g-3}}{\d\phi_0{}^{2g-3}} \, e^{-2\pi\i\m\phi_0-\g\phi_0{}^2}\Big|_{\phi_0=0} \cr
  &= \frac{\i\pi}{(2g-3)!} \int_{-\infty}^{+\infty} \d\phi_0 \; \delta(\phi_0) \; \frac{\d^{2g-3}}{\d\phi_0{}^{2g-3}} \, e^{-2\pi\i\m\phi_0-\g\phi_0{}^2} \cr
  &= -\frac{\i\pi}{(2g-3)!} \int_{-\infty}^{+\infty} \d\phi_0 \; e^{-2\pi\i\m\phi_0-\g\phi_0{}^2} \; \frac{\d^{2g-3}}{\d\phi_0{}^{2g-3}}\,\delta(\phi_0) \;.
\end{align}
By comparing \eqref{EQ:Delta_res} with \eqref{EQ:Q'_as_Fourier_trasnform}, we can fix $\Delta^{\text{res}}_g$. Interestingly, the dependence of $\Delta^{\vphantom{\text{r}}}_g$ on $\m$ turns out to be very mild. In particular, it takes the suggestive form
\begin{align}\label{EQ:Delta_g}
  \Delta^{\vphantom{\text{r}}}_g(\phi_0,\m)
  &= -\frac{1}{(2 g-3)!} \, \frac{\d^{2g-3}}{\d\phi_0{}^{2 g-3}} \left(\pv\frac{1}{\phi_0} + \i\pi\,\sign(\m)\,\delta(\phi_0)\right) \;.
\end{align}

In the next section, we will explain the origin of the one-loop determinant \eqref{EQ:Delta_g} in terms of a localization computation. For the moment, we just want to point out that a one-loop determinant given by a tempered distribution is somewhat unusual in the context of supersymmetric localization. However, in Section~\ref{SEC:toy_model} we will show through a finite-dimensional toy model that this occurs naturally when adopting the particular localization scheme employed in the present paper.

\section{Localization}\label{SEC:localization}

In this section, we compute the contribution to the partition function of 2dYM on a closed surface from each instanton sector using a stationary phase version of supersymmetric localization.
We will follow the approach of Witten \cite{Witten:1992xu}, but work only with the basic multiplet, comprising $A$, $\psi$ and $\phi$. We will not introduce the additional fields forming the $\mathcal{N}=(2,2)$ vector multiplet which were used to localize more general A-twisted models in e.g.\ \cite{Benini:2016hjo}.
Our primary contribution will be a careful reexamination of the effective action for the bosonic moduli, which will yield a reasonable result for the perturbation series of 2dYM. 

We begin with a brief review of non-abelian localization, following \cite{Witten:1992xu}.
We then carry out stationary phase localization in the toy example of non-abelian localization presented in \cite{Witten:1992xu}.
Finally, we carry out the complete computation for 2dYM with gauge group $\mathrm{SU}(2)$.
The computation uses results for Yang--Mills connections collected in Appendix \ref{sec:Yang-Mills-connections}.
The contributions to the partition function match the all-orders perturbative expansion presented in Section \ref{SEC:Z_from_poisson_summation}, up to a factor which is independent of the coupling constant.

\subsection{Non-abelian localization}\label{sec:Non-abelian-localization}
In \cite{Witten:1992xu}, a localization formula is introduced, which is a non-abelian generalization of the Duistermaat--Heckman localization formula in symplectic geometry.

Let $G$ be a Lie group acting on a manifold $\mathscr{M}$. We denote with $\mathfrak{g}$ its Lie algebra and with $V(\phi)$ the fundamental vector field generated by $\phi\in\mathfrak{g}$. The equivariant de Rahm complex $(\Omega^\bullet_G(\mathscr{M}), \mathrm{D})$ is the sequence of smooth $G$-equivariant differential forms\footnote{
  A smooth $G$-equivariant differential form is a smooth map $\alpha:\mathfrak{g}\to\Omega^{\bullet}(\mathscr{M})$ such that the following diagram commutes for any $g\in G$:
  $$
  \begin{tikzcd}[ampersand replacement=\&]
    \mathfrak{g} \arrow[r, "\alpha"] \arrow[d, "\mathrm{Ad}_g"'] \& \Omega^{\bullet}(\mathscr{M}) \arrow[d, "g"] \\
    \mathfrak{g} \arrow[r, "\alpha"']                            \& \Omega^{\bullet}(\mathscr{M})               
  \end{tikzcd}
  $$
} on $\mathscr{M}$, together with the nilpotent twisted differential $\mathrm{D}:\Omega^\bullet_G(\mathscr{M})\to\Omega^\bullet_G(\mathscr{M})$ defined as
\begin{align}
  (\mathrm{D}\alpha)(\phi) = \d\alpha(\phi) - \i \, i_{V(\phi)} \alpha(\phi) \;,
\end{align}
for $\alpha\in\Omega^\bullet_G(\mathscr{M})$. The equivariant de Rahm complex induces the equivariant cohomology group $H^\bullet_G(\mathscr{M})$.

The localization formula of \cite{Witten:1992xu} deals with the equivariant integral
\begin{align}\label{EQ:equivariant_integral}
  \oint : \Omega^\bullet_G(\mathscr{M}) \longrightarrow \mathbb{C}
\end{align}
defined as
\begin{align}
  \oint \alpha = \frac{1}{\vol G}\int_{\mathfrak{g}}\d\phi \; e^{-\g\langle\phi,\phi\rangle} \int_{\mathscr{M}} \alpha(\phi) \;,
\end{align}
where $\d\phi$ is the volume form induced by the Killing form $\langle\cdot,\cdot\rangle$ on $\mathfrak{g}$. The exponential term with $\g>0$ ensures the convergence of the integral. Since the equivariant integral of an equivariantly exact form vanishes, \eqref{EQ:equivariant_integral} descends to a map $H^\bullet_G(\mathscr{M}) \to \mathbb{C}$.

The localization scheme presented in \cite{Witten:1992xu} relies on the introduction of the localizing term $e^{-t\mathrm{D}\lambda}$, where $\lambda\in\Omega^1_G(\mathscr{M})$.\footnote{
  To ensure convergence, it is sufficient to choose a $G$-invariant $\lambda$.}
In fact, if $\alpha$ is equivariantly closed,
\begin{align}
  \oint \alpha = \oint \alpha \; e^{-t\mathrm{D}\lambda} \;, 
\end{align}
since both integrands belong to the same cohomology class in $H^\bullet_G(\mathscr{M})$. Now the exponential brings a $e^{\i tK}$ factor, with $K = i_{V(\phi)} \lambda(\phi)$. In the limit where $t\to\infty$, one can use the method of the stationary phase to prove that the integral can be expressed as a set of local contributions with
\begin{align}\label{EQ:localization_formula}
  \oint \alpha = \sum_{\sigma} Z_\sigma \;,
\end{align}
where $\sigma$ labels the critical sets $\mathscr{M}_\sigma$, i.e.\ the components of the subspace of $\mathscr{M}$ where $\d K = 0$.

This localization scheme is applied to the context of symplectic geometry: suppose $\mathscr{M}$ is a symplectic manifold with symplectic form $\omega$ while $G$ acts as a symplectomorphism with momentum map $\mu$, i.e.\ $\mu:\mathscr{M}\to\mathfrak{g}^*$ is such that $i_{V(\phi)}\omega = -\d\mu(\phi)$.
Let us denote $I = \langle\mu,\mu\rangle$. We choose
\begin{align}\label{EQ:lambda_symplectic}
  \lambda = \frac{1}{2}(\d I)J \;,
\end{align}
where $J$ is an almost complex structure $J$ compatible with $\omega$, namely
\begin{align}
  g(\cdot, \cdot) = \omega(\cdot, J(\cdot)) \;.
\end{align}
Because, as shown in \cite{Witten:1992xu}, the critical points of $K$ coincide with the critical points of $I$, in this localization scheme, the integral localizes on the stationary points of $\langle\mu,\mu\rangle$. Clearly, the subspace $\mu^{-1}(0)$ is the absolute minimum of such a function.

We still have to specify which form we wish to integrate: this is the unique equivariant extension of the symplectic form that is equivariantly closed, namely
\begin{align}
  \xi(\phi) = \omega - \i\mu(\phi) \;.
\end{align}
In \cite{Witten:1992xu} it is shown that, when restricted to a tubular neighborhood $\mathscr{N}_\sigma$ of $\mathscr{M}_\sigma$, the equivariant integral of $e^\xi$ is proportional to
\begin{align}
  \int_{\mathscr{N}_\sigma} e^{\omega + t\d\lambda + W/(2\g)} \;.
\end{align}
Here, $W$ is a scalar that has a local minimum at $\mathscr{M}_\sigma$, where $W=I_\sigma$, the constant value of $I$ over $\mathscr{M}_\sigma$. As a consequence, in the large-$t$ limit, we find that each localized contribution to the equivariant integral is exponentially suppressed as
\begin{align}\label{EQ:higher_critical_points_asymptotics}
  Z_\sigma \sim e^{-I_\sigma/(2\g)} \;,
\end{align}
with the only exception being the global minimum $\mu^{-1}(0)$ where $I$ vanishes. The same conclusions apply straightforwardly to more general integrals of the form
\begin{align}
  \oint \beta \, e^{\xi} \;,
\end{align}
where $\beta$ has at most a polynomial dependence on $\phi$.

\subsection{An abelian toy model}\label{SEC:toy_model}
In the appendix of \cite{Witten:1992xu}, the following example is considered. Let $\mathscr{M}\simeq S^2$ with the usual symplectic volume form $\omega = \sin\theta \, \d\theta\wedge\d\psi$, where $0\leq\theta\leq\pi$ and $0\leq\psi<2\pi$ are the conventional spherical coordinates. We then consider the $\mathrm{U}(1)$ group action generating azimuthal translations. By endowing the Lie algebra $\mathfrak{u}(1)$ with the standard metric such that $\vol \mathrm{U}(1) = 2\pi$, we identify $\mathfrak{u}(1)^* \simeq \mathfrak{u}(1) \simeq \mathbb{R}$, the unit generator on $\mathfrak{u}(1)$ being associated with the fundamental vector field $V = \partial_\psi$. Such a group action on $(S^2,\omega)$ is a symplectomorphism and the associated moment map reads $\mu = -\cos\theta$, up to constant terms. In this setup, the equivariant integral amenable to localization reads 
\begin{align}\label{EQ:Z_toy_model}
  Z = \frac{1}{2\pi} \int_{-\infty}^\infty\frac{\d\phi}{2\pi} \int_{S^2} e^{\omega - \i\phi\mu - \g\phi^2/2} \;.
\end{align}

We now want to revisit this example and show in detail how to explicitly compute the contributions of each critical point of $I$ to the integral. We start by computing the localizing term
\begin{align}\label{EQ:eta_localizing}
  \eta(\phi)
  &= e^{-t\mathrm{D}\lambda} \cr
  &= e^{-t\,\d\lambda+\i t\phi\;i_V\!\lambda} \;,
\end{align}
where
\begin{align}
  \lambda &= \sin^2\!\theta\,\cos\theta\;\d\psi
\end{align}
is induced by the almost complex structure on $S^2$,
\begin{align}
  J = \frac{1}{\sin\theta} \, \partial_\psi\otimes\d\theta - \sin\theta \, \partial_\theta\otimes\d\psi \;,
\end{align}
as in \eqref{EQ:lambda_symplectic}.
For convenience, we rewrite it as $\eta = e^{\eta_0 + \eta_2}$, where
\begin{align}
  \eta_0 &= \i t\phi\,\sin^2\!\theta\,\cos\theta \;, \cr
  \eta_2 &= -t \, \partial_\theta(\sin^2\!\theta\,\cos\theta)\,\d\theta\wedge\d\psi \;.
\end{align}
The integral
\begin{align}\label{EQ:deformed_integral_toy_model}
  \frac{1}{2\pi}\int_{S^2} \eta(\phi) \; e^{\omega + \i\phi\cos\theta - \g\phi^2/2}
\end{align}
should receive contributions from the critical points of $I = \cos^2\theta$, i.e.\ from the minimum at $\theta = \pi/2$ and from the two maxima at $\theta = 0$ and $\theta = \pi$.

In a tubular neighborhood of the circle at $\theta = \pi/2$,
\begin{align}
  \eta_0\big|_{\theta\sim\pi/2} &\sim -\i t\phi\,z \;, \\
  \eta_2\big|_{\theta\sim\pi/2} &\sim t\,\d z\wedge\d\psi \;,
\end{align}
where $z = \theta -\pi/2$. As the integral \eqref{EQ:deformed_integral_toy_model} localizes for $t\to\infty$, the contribution of the fluctuations about the minimum at $\theta = \pi/2$ is captured by an integral over the normal bundle, where we have expanded at the leading order as above:
\begin{align}
  \Delta_{\frac{\pi}{2}\!}(\phi)
  &= t \int_0^{2\pi} \frac{\d\psi}{2\pi} \int_{-\infty}^\infty \d z \; e^{-\i t\phi z} \;.
\end{align}
We interpret the integral over $z$ as the Fourier transform of a constant tempered distribution and write
\begin{align}
  \Delta_{\frac{\pi}{2}\!}(\phi)
  &= 2\pi\,\delta(\phi) \;.
\end{align}

We are left with the contribution coming from the higher critical points. The chart $(\theta, \psi)$ is singular on the maxima, so in a neighborhood of $\theta = 0$ we define $x = \theta \cos\psi$ and $y = \theta \sin\psi$, in terms of which
\begin{align}
  \eta_0\big|_{\theta\sim0} &\sim \i t\phi\,(x^2+y^2) \;, \\
  \eta_2\big|_{\theta\sim0} &\sim -2t \, \d x\wedge\d y \;.
\end{align}
The fluctuation integral over the tangent space at $\theta = 0$ reads
\begin{align}\label{EQ:fluctuation_integral_maximum}
  \Delta\mathmakebox[\widthof{${}_{\frac{\pi}{2}\!}$}][c]{{}_0}(\phi) &= -\frac{t}{\pi} \int_{\mathbb{R}^2} \d x \, \d y \; e^{\i t\phi\,(x^2+y^2)} \;.
\end{align}
This converges to $-\i/\phi$ in the half-plane $\Im \phi>0$ (given $t>0$). In order to obtain a sensible result for real $\phi$ we must interpret again the integral in a distributional sense. This is achieved by taking the limit for $\Im\phi\to0^+$ and using the distributional identity
\begin{align}\label{EQ:pv_as_lim}
  \lim_{\varepsilon\to0^+} \frac{1}{x\pm\i\varepsilon} = \pv \frac{1}{x} \mp\i\pi\,\delta(x) \;.
\end{align}
In fact, upon switching to polar coordinates one recognizes the expression for the Fourier transform of the Heaviside step function.\footnote{\label{FN:Heitler}
  In a distributional sense,
  \begin{align*}
    \int_{0}^{\infty}\d x \; e^{\mp2\pi\i px} = \pm\frac{1}{2\pi\i}\pv\frac{1}{p} + \frac{\delta(p)}{2} \;.
  \end{align*}
  The above goes also under the name of Heitler function.
} The result is the tempered distribution
\begin{align}
  \Delta_0(\phi)
  &= -\i\pv\frac{1}{\phi} - \pi\,\delta(\phi) \;.
\end{align}
We will not repeat the analysis for the other maximum at $\theta=\pi$ as it is completely analogous and simply amounts to replacing $\phi\mapsto-\phi$. This gives
\begin{align}
  \Delta_{\pi}(\phi) &= +\i\pv\frac{1}{\phi} - \pi\,\delta(\phi) \;.
\end{align}

According to \eqref{EQ:localization_formula}, we can write the integral \eqref{EQ:Z_toy_model} as
\begin{align}
  Z = Z_0 + Z_\frac{\pi}{2} + Z_\pi \;,
\end{align}
where
\begin{align}
  Z_\theta = \int_{-\infty}^{\infty} \frac{\d\phi}{2\pi} \; \Delta_\theta(\phi) \; e^{+\i\phi\cos\theta-\g\phi^2/2} \;.
\end{align}
A direct computation gives
\begin{align}
  Z\mathmakebox[\widthof{${}_{\frac{\pi}{2}\!}$}][c]{{}_0} &= -\frac{1}{2} \erfc\!\bigg(\frac{1}{\sqrt{2\g}}\bigg) \;, \\
  Z_\frac{\pi}{2} &=  1 \;, \\
  Z\mathmakebox[\widthof{${}_{\frac{\pi}{2}\!}$}][c]{{}_\pi} &= -\frac{1}{2} \erfc\!\bigg(\frac{1}{\sqrt{2\g}}\bigg) \;.
\end{align}
As expected, the contributions associated with the two maxima are exponentially suppressed as $\g\to0^+$. In fact, we find
\begin{align}
  Z_{0,\pi} \sim e^{-1/(2\g)} \;,
\end{align}
consistently with \eqref{EQ:higher_critical_points_asymptotics}. Combining the three terms produces the correct result
\begin{align}
  Z = \erf\!\bigg(\frac{1}{\sqrt{2\g}}\bigg) \;.
\end{align}

\subsection{The Yang--Mills effective action}

We will now show that a supersymmetric localization computation of the partition function of 2d Yang--Mills yields results for the perturbation series which are in line with those derived in previous Section~\ref{SEC:Z_from_poisson_summation}.
To begin, we provide in Table~\ref{tab:dictionary} a dictionary for translating the main components of non-abelian localization, as reviewed in Section~\ref{sec:Non-abelian-localization}, to the language of 2d supersymmetric gauge theory.
In this section we use the notation and results from Appendix~\ref{sec:Yang-Mills-connections}.
We will implicitly work only with $g\ge2$. The localization computation for $g=0$ and $g=1$ can be performed with standard methods.

\begin{table}
\begin{tabulary}{\linewidth}{CC}
  \toprule
  \textbf{Symplectic Geometry} & \textbf{Supersymmetric 2d Yang--Mills} \\
  \bottomrule
  the Lie group $\mathcal{G}$ &
  the gauge group: the automorphism group of a fixed principal $G$-bundle over $\Sigma$ \\\hline
  the smooth $\mathcal{G}$-manifold $X$ &
  the space of connections $\mathcal{A}$ on a fixed principal $G$-bundle over $\Sigma$ \\\hline
  the action of $\mathcal{G}$ on $X$ &
  gauge transformations acting on $\mathcal{A}$ \\\hline
  coordinates on the Lie algebra of $\mathcal{G}$ &
  the field $\phi$ as an adjoint-valued scalar on $\Sigma$ \\\hline
  coordinates $\{x\} $ on $X$ &
  connections $A\in\mathcal{A}$ \\\hline
  $\partial_{x}$ as an element of the tangent space $TX$ &
  $\delta A$ as an adjoint-valued one-form on $\Sigma$ \\\hline
  a one-form dual to $\delta A$ (as an element of the exterior algebra) &
  an adjoint-valued fermion $\psi$ which is also a one-form on $\Sigma$ \\\hline
  an equivariant differential form (in local coordinates) &
  a gauge invariant superfield $\Phi(A,\psi,\phi)$ \\\hline
  the exterior product &
  the product of superfields \\\hline
  equivariant integration of equivariant differential forms &
  supersymmetric path integral \\\hline
  the equivariant differential $D$ &
  the supersymmetry variation $\delta$ \\\hline
  the symplectic form $\omega$ &
  the Atiyah--Bott symplectic form $\omega_{\text{AB}}$ \\\hline
  equivariantly closed forms &
  $\delta$ closed superfields \\\hline
  equivariant cohomology of $X$ &
  the $\delta$ cohomology \\\hline
  the equivariant extension $\omega-\i\mu$ &
  the superfield $\frac{1}{2}\psi\wedge\psi+\i\phi F$ \\\bottomrule
\end{tabulary}
\caption{The dictionary for translating non-abelian localization in symplectic geometry into supersymmetric gauge theory.\label{tab:dictionary}}
\end{table}

\paragraph{Localization of 2dYM.}
As explained in Section~\ref{SEC:2dYM}, 2dYM has a completely equivalent formulation in terms of fields $A,\psi$, $\phi$, and the supersymmetric action \eqref{eq:cohomological_Yang_Mills_action} that we reproduce here for convenience:
\begin{align}\label{eq:cohomological_Yang_Mills_action_2}
  S_{\text{SUSY}}[A,\psi,\phi] = -\frac{1}{4\pi^2}\int\tr\Big(\i\phi F + \frac{1}{2}\,\psi\wedge\psi\Big) -\frac{\g}{8\pi^4}\int \vf\,\tr(\phi^{2}) \;.
\end{align}
The supersymmetry acting on the fields is generated by
\begin{align}
  \delta A &= \i\psi \;, &
  \delta\psi &= -\d_{A}\phi \;, &
  \delta\phi &= 0 \;.
\end{align}
Note that
\begin{align}
  \delta F &= \i\,\d_{A}\psi \;.
\end{align}

We would like to compute the partition function for this theory on a closed genus $g$ surface $\Sigma$ of unit area, using localization.
To begin, we deform the 2dYM action using a $\delta$-exact term.
Our choice of $\delta$-exact term appears in \cite{Witten:1992xu} at a certain stage in the computation (cf.\ Eq.\ (3.30) of \cite{Witten:1992xu}):\footnote{
  See \ref{sec:Yang-Mills-connections} for the precise definitions of $\d_{A}$ and $\hat{F}$.
}
\begin{align}\label{eq:localizing_action}
  S_{\text{loc.}}
  &= t\;\delta \bigg( \frac{1}{32\pi^{4}} \int_{\Sigma}\tr\bigg(\i\psi\wedge \d_{A}^{\dagger}F\bigg)\bigg) \cr
  &= \frac{t}{32\pi^{4}} \int_{\Sigma} \tr\bigg(\i\phi\wedge \d_{A}^{\vphantom{\dagger}}\d_{A}^{\dagger}F + \psi\wedge\bigg(\frac{1}{2}\Delta_{A}^{\vphantom{\dagger}}+\hat{F}\bigg)\psi\bigg) \;.
\end{align}
By a standard argument, the $\delta$-exact deformation does not change the expectation value of the partition function and of $\delta$-closed observables. The precise coefficient used here was chosen for later convenience.

The bosonic part of the action \eqref{eq:localizing_action} is purely imaginary, and the localization we have in mind is an exact stationary
phase approximation.\footnote{
  Note that this reality of the bosonic action would not be true in the $\mathcal{N}=(2,2)$ version of this theory, where $\phi$ is a \emph{complex} scalar.
}
This should be contrasted with the usual setup for supersymmetric localization, as expounded for instance in \cite{Schwarz:1995dg}, whereby a real positive-semi-definite localizing action is used and the result is an exact saddle point approximation.
Our first assumption is that such a localization is still well defined, an assumption which is borne out by the results in \cite{Witten:1992xu}, albeit without a derivation of the finite dimensional model.
Note also that this type of localizing term is the direct analogue of the one used in the toy example, namely \eqref{EQ:eta_localizing}.

According to \cite{Witten:1992xu}, the localizing action \eqref{eq:localizing_action} also brings in configurations from infinity, for any finite value of $t$, resulting in expressions for the expectation values of $\delta$-closed observables at finite $\g$ which are different from what one could have obtained via a standard localizing term leading to a cohomological TQFT.
The additional configurations are unstable instantons, exponentially suppressed at small $\g$, allowing them to be neatly separated
from the topological data at $\g\rightarrow0$.
This latter fact, which was central to the arguments leading to the computation of the intersection numbers on the moduli space of flat connections performed in \cite{Witten:1992xu}, is not relevant for our purposes.
We \emph{define} the trans-series type expansion in $\g$ which we wish to compute as the one arising from our 2dYM action at any
value of $t$.
Indeed, for the purpose of computing the partition function, the theory we consider is equivalent to \emph{physical} 2dYM at any value of $t$.

The equations of motion arising from \eqref{eq:localizing_action} imply that the stationary phase configurations satisfy
\begin{align}
  \d_{A}^{\vphantom{\dagger}}\d_{A}^{\dagger}F = 0 \;.
\end{align}
Using the inner product, one sees that this equation is equivalent to the Yang--Mills equation 
\begin{align}\label{eq:Yang-Mills_equation}
  \d_{A}^{\dagger}F = 0 \;.
\end{align}
Treating the fermions as perturbations, the equations of motion for $A$ are equivalent to 
\begin{align}
  \Delta_{A}\d_{A}\phi &= 0 \;,
\end{align}
which implies that $\d_{A}\phi$ is harmonic and therefore $\d_{A}^{\dagger}\d_{A}^{\vphantom{\dagger}}\phi$ vanishes leading to
\begin{equation}\label{eq:eom_phi}
  \d_{A}\phi = 0 \,.
\end{equation}
Solutions to this equation parameterize the Lie algebra of the isotropy group of $A$, i.e.\ continuous families of gauge transformations left
unbroken by $A$.
This is a finite dimensional Lie algebra. The connection $A$ is called irreducible if \eqref{eq:eom_phi} has no non-trivial solutions.

In order to perform the path integral, we will also need a BRST gauge fixing multiplet with the standard fields $c,\bar{c}$ and $B$, and a BRST transformation $\delta_{\text{BRST}}$ acting as
\begin{align}
  \delta_{\text{BRST}}c &= cc \;, &
  \delta_{\text{BRST}}\bar{c} &= \i B \;, &
  \delta_{\text{BRST}}B &= 0\;.
\end{align}
The action of $\delta_{\text{BRST}}$ on the rest of the fields is equivalent to a gauge transformation with parameter $c$, for instance\footnote{
  The covariant derivative acting on the ghost should really be with respect to the total connection $A+a$, but in the absence of ghost zero modes the operator $\d_{A}$ is what remains after localization.
}
\begin{align}
  \delta_{\text{BRST}}a = -\d_{A}c \;.
\end{align}
Localization in the presence of BRST symmetry requires consideration of the combined complex with a transformation $\delta+\delta_{\text{BRST}}$, however the localizing term we used is $\delta_{\text{BRST}}$ invariant.

We can impose background gauge by taking
\begin{align}\label{eq:gauge_fixing_term}
  S_{\text{g.f.}}
  &= (\delta+\delta_{\text{BRST}})\int_{\Sigma}\vf\;\tr\Big(\bar{c}\,\d_{A}^{\dagger}a\Big) \cr
  &= \int_{\Sigma}\vf\;\tr\left(\i B\,\d_{A}^{\dagger}a+\bar{c}\,\Delta_{A}^{\vphantom{\dagger}}c-\i\bar{c}\,\d_{A}^{\dagger}\psi\right)\,.
\end{align}
In these expressions, $A$ denotes a fixed background connection satisfying the Yang--Mills equation, while $a$ denotes the fluctuations in the connection orthogonal to $A$ under the standard inner product.
Both $a$ and $\psi$ are thus adjoint-valued one-forms. 

In order to simplify matters, we begin the calculation of the partition function by integrating out the BRST auxiliary field $B$.
This imposes background gauge on $a$.
We implicitly use this condition to simplify various expressions involving $a$ in what follows.
Integration over $B$ yields a functional delta function setting all modes of $a$ which do not satisfy background gauge to zero, and in addition gives the following functional determinant:
\begin{align}
  \mathcal{N}_{\text{BRST}} &= (\det \d_{A}^{\dagger})^{-1} \;.
\end{align}

\paragraph{The bosonic moduli space.}

At this point, we specialize to $G=\mathrm{SU}(2)$.
We will comment later regarding the challenges involved for higher rank groups.
We will use a basis for $\mathfrak{su}(2)$ consisting of
\begin{align}
  \tau_{i} = \i\sigma_{i} \;,
\end{align}
where $\sigma_{i}$ are the Pauli matrices. Note that
\begin{align}
  \tr(\tau_{i}\tau_{j}) = -2\delta_{ij} \;.
\end{align}

As mentioned in Section~\ref{SEC:Z_from_poisson_summation}, the solutions to the equation of motion \eqref{eq:Yang-Mills_equation} are Yang--Mills connections which can carry magnetic flux \cite{Atiyah:1982fa}.
The flux is parameterized by an element of the magnetic weight lattice of $G$, denoted $\Gamma$, which for $G=\mathrm{SU}(2)$ is a single integer $\m$, such that $\star F = 2\pi\m\,\tau_{3}$.

For $\m=0$, the solutions modulo gauge transformations form the moduli space of flat $\mathrm{SU}(2)$ connections on $\Sigma$, whose irreducible part for $g\ge2$ has dimension $6g-6$.
For $\m\ne0$, the moduli space consists of reducible connections, and is equivalent to a $2g$ dimensional space of flat $\mathrm{U}(1)$ connections that we denote by $\mathcal{M}_\m$.
For $g=1$, only the reducible part appears regardless of $\m$, while for $g=0$, the flux number $\m$ completely characterizes the solutions.

We henceforth restrict to $\m\ne0$, and therefore to reducible connections.
Solutions to the equation of motion for $\phi$ parameterize infinitesimal unbroken gauge transformations.
On the reducible part of the moduli space the isotropy group is $\mathrm{U}(1)$.
A solution is simply a constant profile for the scalar, which we parameterize as $\phi = 2\pi^2\phi_0\,\tau_{3}$, and which is moreover in the same Cartan subalgebra as the magnetic flux or the holonomies of $A$.

Thinking of the $\m$, $\phi_{0}$, and the abelian connections as living in a Cartan subalgebra, specifically the Lie algebra of the maximal torus for the unbroken $\mathrm{U}(1)$, the space of reducible solutions must also be divided by the Weyl group, which acts on the flat connections and on $\phi_{0}$ as a reflection in the Lie algebra of $\mathrm{U}(1)$, as well as taking $\m$ to $-\m$. 

We are interested in evaluating the effective action for $\phi_{0}$ in a sector where $\m\ne0$, and deriving from it the perturbation series in $\g$. The analogous calculation for higher rank groups is the effective action in a sector where $\m\in\Gamma$ lies in the interior of a Weyl chamber, rather than on a boundary.
Such an $\m$ can also be characterized as not being fixed by any element of the Weyl group.
The manipulations below, including the characterization of the various zero modes that arise in the process of localization, is valid for such an $\m$.
For $\m\in\Gamma$ not satisfying this criterion, we must deal with the various branches of the moduli space described above. We will not attempt this here.

\paragraph{Bosonic fluctuations.}

We now expand the bosonic part of \eqref{eq:localizing_action} around the Yang--Mills connection $A$ and the covariantly constant scalar $\phi_{0}$.
The fields $\phi$ and $a$ are fluctuations around this background, orthogonal to the bosonic moduli space in the standard inner product.
The localizing action in background gauge, expanded to quadratic order, and with the usual re-scaling of the fields to eliminate the dependence on $t$, reads
\begin{align}
  S_{\text{loc.}}
  = \frac{1}{32\pi^{4}} \int_{\Sigma} \tr(\i\phi_{0}\wedge K(a)+\i\phi\wedge \d_{A}\Delta_{A}a) + O(t^{-1/2})\,,
\end{align}
where $K(a)$ is an expression quadratic in $a$. 

We now take advantage of the linear dependence of this action on the fluctuation $\phi$, which will not appear elsewhere, to integrate out $\phi$.
This introduces a functional delta function
\begin{equation}\label{eq:delta_function_for_a}
  \delta(\d_{A}\Delta_{A}a) \;.
\end{equation}
Integration of $a$ using the delta function is, as usual, accompanied by a functional determinant
\begin{align}
  \mathcal{N}_{\text{bos.}} = \det{}'\left(\frac{1}{32\pi^{4}}\,\d_{A}\Delta_{A}\right)^{-1}\,,
\end{align}
where $\det'$ denotes a determinant over all modes outside the kernel of \eqref{eq:delta_function_for_a}.
Note that the factor $\mathcal{N}_{\text{bos.}}$ is independent of $\phi_{0}$, but depends explicitly on the background connection $A$.

Integration over the fluctuations $a$ is thus trivial for all modes not in the kernel of the operator $\d_{A}\Delta_{A}$, equivalently the kernel of $\Delta_{A}$.
The kernel itself consists of adjoint-valued one-forms satisfying
\begin{align}
  \d_{A}^{\vphantom{\dagger}}a = \d_{A}^{\dagger}a = 0 \;.
\end{align}
Solutions to these equations can be organized according to their $\hat{F}$ eigenvalue (see Appendix~\ref{sec:Yang-Mills-connections} for details).
We denote this eigenvalue $\pm\i\lambda$ following \cite{Atiyah:1982fa}, and noting that in our conventions
\begin{align}
  \lambda = \pm2\pi\m \,.
\end{align}
Vanishing $\lambda$ implies a direction tangent to the moduli space.
For a fixed non-zero $\lambda$, \eqref{eq:harmonic_one_forms} gives the dimension of the kernel as
\begin{align}
  \dim\text{Ker}_{\text{ad}_{\lambda}(P)\otimes\Omega^{1}}(\Delta_{A})=\dim H_{\d_{A}''}^{1}(\Sigma,\text{ad}_{\lambda}(P))+\dim H_{\d_{A}''}^{1}(\Sigma,\text{ad}_{-\lambda}(P))\,.
\end{align}
When evaluated on the kernel, the expression $K(a)$ collapses to $a\wedge\hat{F}a$.
Taking both values of $\lambda$ into account, we get a count of the number of \emph{real} bosonic modes in the kernel of \eqref{eq:delta_function_for_a}
\begin{align}
  N_{b} \equiv 2\dim H_{\d_{A}''}^{1}(\Sigma,\text{ad}_{\lambda}(P)) + 2\dim H_{\d_{A}''}^{1}(\Sigma,\text{ad}_{-\lambda}(P)) \;.
\end{align}

The modes collectively denoted by $a$ will not show up in the fermionic part of the action, and we perform the integration over them here.
Let $\{a_{i}\}_{i=1}^{N_{b}}$ be and orthonormal basis for the modes in the kernel of \eqref{eq:delta_function_for_a}, $x^{i}$ auxiliary bosonic variables, and $x^{i}a_{i}$ the orthogonal projection of the field $a$ to this sector.
The integration over all fluctuations $a$ gives the one-loop determinant
\begin{align}\label{eq:bosonic_one_loop}
  Z_{\text{1-loop}}^{\text{bos.}}
  &= \int\mathcal{D}a \; \delta(\d_{A}\Delta_{A}a) \; \exp\bigg(\frac{1}{32\pi^{4}}\int_{\Sigma}\tr\Big(2\pi^{2}\i\phi_{0}\tau_{3}\Big(a\wedge\hat{F}a\Big)\Big)\bigg) \cr
  &= \mathcal{N}_{\text{bos.}} \int\prod_{i=1}^{N_{b}} \bigg(\d x_{i} \; e^{-2\pi\i\m\phi_{0}x_{i}^{2}}\bigg) \;.
\end{align}
The $N_b$ integrals can be turned into a single radial integral, and further simplified by changing the integration variable to be the square root of the radial coordinate. Moreover, the resulting Jacobian can be expressed in terms of derivatives in $\phi_0$. Namely,
\begin{align}
  \int\prod_{i=1}^{N_{b}} \bigg(\d x_{i} \; e^{-2\pi\i\m\phi_{0}x_{i}^{2}}\bigg)
  = \frac{\Omega_{N_{b}-1}}{2} \, \bigg({-}\frac{1}{2\pi\i\m}\frac{\d}{\d\phi_{0}} \bigg)^{\!N_{b}/2-1} \int_0^\infty \d r \; e^{-2\pi\i\m\phi_{0}r} \;,
\end{align}
where
\begin{align}
  \Omega_{N_{b}-1} = \frac{2\pi^{N_{b}/2}}{\Gamma(N_{b}/2)}
\end{align}
is the volume of the unit ($N_{b}-1$)-sphere.
The last integral is of the same kind encountered in the finite-dimensional toy model of Section~\ref{SEC:toy_model} (cf.\ Footnote~\ref{FN:Heitler}). The result reads
\begin{align}
  Z_{\text{1-loop}}^{\text{bos.}}
  &= -\frac{\mathcal{N}_{\text{bos.}}}{(-2\i\m)^{N_{b}/2}\,\Gamma(N_{b}/2)} \, \frac{\d^{N_{b}/2-1}}{\d\phi_{0}{}^{N_{b}/2-1}}\,\bigg(\pv\frac{1}{\phi_{0}} + \i\pi\sign(\m)\,\delta(\phi_{0})\bigg) \;.
\end{align}

\paragraph{Fermionic fluctuations.}

We now examine the fermionic part of the action
\begin{align}
  \int_{\Sigma} \tr\left(\frac{1}{32\pi^{4}}\,\psi\wedge\left(\frac{1}{2}\Delta_{A}+\hat{F}\right)\psi + \bar{c}\,\Delta_{A}c - \i\bar{c}\,\d_{A}^{\dagger}\psi\right) \;.
\end{align}
Suppose that we substitute for $A$ a solution of the Yang--Mills equations.
The operator $\Delta_{A}$ acting on scalars is positive semi-definite.
The only modes of $c$ or $\bar{c}$ in the kernel of $\Delta_{A}$ are those with $\Lambda$ eigenvalue $0$.
These are covariantly constant scalars in the same Cartan subalgebra as $\star F$.
It is possible to further gauge fix to get rid of such modes, but this turns out to be equivalent to simply discarding them from the action a priori, which is the approach we follow.\footnote{
  We have implicitly done the same with the equivalent mode of the field $B$, which decouples from the gauge fixing term \eqref{eq:gauge_fixing_term}. Hence, there are no zero modes in the ghost sector.
} 

Zero modes of $\psi$ lie in the kernel of the operator $\frac{1}{2}\Delta_{A}+\hat{F}$.
These come in two varieties:
\begin{enumerate}
  \item Modes which are simultaneously in the kernel of both $\hat{F}$ and $\Delta_{A}$.
  These are in one to one correspondence with ordinary harmonic one-forms, and comprise a vector space of dimension $2g$.
  They are physical zero modes, representing tangent vectors to the moduli space of flat connections \cite{Atiyah:1982fa}.
  We will denote them $\{\psi_{0}^{i}\}_{i=1}^{2g}$.\footnote{
    Such modes should be completely absent from the non-classical part of the effective action (cf.\ \cite{Benini:2016hjo}), although we could not confirm this. We will only show that the partition function can be recovered without taking such terms into account.
  } 
  \item Modes for which $\hat{F}\psi$ is equal to $-|\lambda|\psi$ and $\Delta_{A}\psi$ takes its minimal value of $2|\lambda|$.
  These are spurious zero modes which must be dealt with before continuing with the calculation.
  According to \eqref{eq:minimal_eigenvalue_one_forms} all such modes are of the form $\d y$ for some scalar $y$.
  Moreover, for the $\{\lambda,-\lambda\} $ pair we have the following number of \emph{real} fermionic zero modes of this type
  \begin{align*}
    N_{f} \equiv 2\dim H_{\d_{A}''}^{0}(\Sigma,\text{ad}_{\lambda}(P)) + 2\dim H_{\d_{A}''}^{0}(\Sigma,\text{ad}_{-\lambda}(P)) \;.
  \end{align*}
\end{enumerate}
Note that modes of type 2 are not by themselves zero modes of the full fermionic action.
Because of the mixing term between $\psi$ and $\bar{c}$, the actual zero modes consist of the modes
\begin{align}
  \psi &= \d_{A}y \;, &
  c &= \i y \;.
\end{align}
The quadratic form for the fermions being antisymmetric, the zero modes are orthogonal to all non-zero modes and can be considered separately.
The non-zero modes, meanwhile, contribute a factor of 
\begin{align}
  \mathcal{N}_{\text{ferm.}} = \det{}'\bigg(\frac{1}{32\pi^{4}}\bigg(\frac{1}{2}\Delta_{A}+\hat{F}\bigg)\bigg)\,\det{}'(\Delta_{A}) \;.
\end{align}
As with the factor $\mathcal{N}_{\text{bos.}}$, this factor is independent of $\phi_{0}$ but depends explicitly on $A$.

In order to lift the type 2 zero modes, we must introduce an additional exact term.
Let $\{ y^{i}\} _{i=1}^{N_{f}}$ be an orthonormal basis for the bosonic version of the type 2 zero modes.
Let $v_{i}$ be anti-commuting parameters such that $v_{i}y^{i}$ is the orthogonal projection of the ghost field $c$ to the zero mode sector. We choose to add
\begin{align}\label{eq:femion_lifting_term}
  (\delta+\delta_{\text{BRST}}) \; \frac{1}{16\pi^{3}} \int_{\Sigma} \tr\big({-}\i F\big[\phi_{0},v_{i}y^{i}\big]\big)
  &= \frac{1}{16\pi^{3}}\int_{\Sigma}\tr\big(\d_{A}\psi\,\big[\phi_{0},v_{i}y^{i}\big]\big) \cr
  &= -\frac{\i}{16\pi^{3}}\int_{\Sigma}\tr\big(\d_{A}\d_{A}\big(v_{j}y^{j}\big)\big[\phi_{0},v_{i}y^{i}\big]\big) \cr
  &= \frac{\i}{16\pi^{3}}\,v_{j}v_{i}\int_{\Sigma}\tr\big(\big[F_{0},y^{j}\big]\big[\phi_{0},y^{i}\big]\big) \;,
\end{align}
where in the second line we have already taken the quadratic approximation and expanded $\psi$ in the appropriate modes. Evaluating the resulting Pfaffian in the variable $v$ gives a factor of
\begin{equation}\label{eq:one_loop_fermion_contribution}
  (2\i\m\phi_{0})^{N_f/2} \;.
\end{equation}

\paragraph{The effective action.}

The classical part of the effective action for the moduli is obtained by evaluating \eqref{eq:cohomological_Yang_Mills_action_2} on the zero modes $\phi_{0}^{\vphantom{i}},\psi_{0}^{i}$.
Because \eqref{eq:cohomological_Yang_Mills_action_2} depends only on $F$ and not on $A$, it does not depend on the flat connections, and reduces to a single expression on every connected component of the moduli space of Yang--Mills connections. This gives
\begin{align}
  -\log Z_{\text{effective}}^{\text{cl.}\vphantom{()}} = \g\phi_{0}^{2} + 2\pi\i\m\phi_{0}^{\vphantom{i}} + \frac{1}{2}\psi_{0}^{i}\wedge\psi_{0i}^{\vphantom{i}} \;.
\end{align}

The quantum part of the effective action for the mode $\phi_{0}$ is the distribution defined by the product of \eqref{eq:bosonic_one_loop} and \eqref{eq:one_loop_fermion_contribution}.
Using the expressions for $N_{b}$ and $N_{f}$, the Riemann--Roch theorem, and fact that
\begin{align}
  c_{1}(\text{ad}_{-\lambda}) = -c_{1}(\text{ad}_{\lambda}) \;,
\end{align}
and 
\begin{align}
  \dim_{\mathbb{C}} \text{ad}_{\lambda} = 1 \;,
\end{align}
we see that
\begin{align}
  \frac{1}{2}\left(N_b-N_f\right) = 2g-2 \;.
\end{align}
Moreover, the factors of $\phi_{0}$ multiplying the distribution give a different distribution with the same singular support.

The complete expression for the partition function then reads
\begin{align}
  Z = \frac{1}{2}\sum_{\m\in\mathbb{Z}}\int \d\phi_{0}^{\vphantom{i}}\int_{\mathcal{M}_{\m}}\int\prod_{i}\d\psi_{0}^{i} \;\, Z_{\text{effective}} \;,
\end{align}
where the full effective action is given by
\begin{align}
  Z_{\text{effective}}^{\vphantom{()}} = Z_{\text{effective}}^{\text{cl.}\vphantom{()}} \, Z_{\text{effective}}^{(1)} \, Z_{\text{effective}}^{(2)} \, Z_{\text{effective}}^{(2)} \;.
\end{align}
Taking into account both the measure for $\phi_{0}$ and the division by the center of $\mathrm{SU}(2)$, the final result for the quantum part of the effective action has been written as a product of the three expressions
\begin{align}\label{eq:quantum_effective_action}
  Z_{\text{effective}}^{(1)}
    &= -\frac{1}{(2g-3)!} \, \frac{\d^{2g-3}}{\d\phi_{0}{}^{2g-3}}\,\bigg(\pv\frac{1}{\phi_{0}} + \i\pi\sign(\m)\,\delta(\phi_{0})\bigg) \;, \\
  Z_{\text{effective}}^{(2)}
    &= (-4\m^2)^{1-g} \;, \\
  Z_{\text{effective}}^{(3)}
    &= \pi^{2}\mathcal{N}_{\text{BRST}}\,\mathcal{N}_{\text{bos.}}\,\mathcal{N}_{\text{ferm.}} \;. \vphantom{\bigg)^{2g-3}}
\end{align}

The first expression, the only one which depends on $\phi_{0}$, is indeed the one-loop determinant \eqref{EQ:Delta_g} predicted by the Poisson summation procedure.
We have separated the constants so that the result matches without prefactors.
This is our primary result from this section. 

The second expression amounts to an allowed renormalization of the type described in \eqref{eq:allowed_renormalization} albeit depending on the flux number $\m$.
The third expression is independent of $\phi_{0}$, but depends explicitly on the point in the Yang--Mills moduli space $\mathcal{M}_{\m}$ given by the background connection $A$.
In analogy with the discussion of the analytic torsion in the case of flat connections \cite{Witten:1991we}, we expect the integral of this expression over the moduli space to be a topological invariant, possibly depending on $\m$, and for the product
\begin{align}
  Z_{\text{effective}}^{(2)} \int_{\mathcal{M}_{\m}}Z_{\text{effective}}^{(3)} \;,
\end{align}
together with a numerical factor coming from the fermionic zero modes to be an $\m$ independent renormalization.
It should be possible to show this by manipulation of the determinants with an appropriate regularization scheme, but we have not done so.

To summarize, localization using the purely imaginary localizing term \eqref{eq:localizing_action} immediately yields the correct bosonic moduli space, while the fermion zero modes require a more detailed analysis.
The non-classical part of the effective action for the modulus $\phi_{0}$ comes from integrating out a collection of distinguished bosonic and fermionic fluctuations, the harmonic modes $x$ and the type-$2$ fermionic zero modes $y$.
A priori, counting these modes requires a detailed mode by mode analysis, but the combination appearing in the final expression can be recovered from an index theorem, in this case the Riemann--Roch theorem.
This is also the situation one often encounters when using a positive semi-definite localizing term. 

A number of issue remain to be resolved in this calculation.
The computation of the $\phi_{0}$ independent part of the quantum effective action is incomplete.
Moreover, the lifting of the fermion zero modes necessitated the addition of a term with a seemingly independent coefficient which alters the final result by numerical factors.
We hope to return to these matters.

\section{Discussion}\label{SEC:discussion}
We have demonstrated that the full perturbative expansion of the partition function of Yang--Mills theory on a genus $g$ surface around unstable instantons can be reliably computed in two different ways, at least for gauge group $\mathrm{SU}(2)$.
We first showed that the expansion can be computed starting with the lattice gauge theory result for the partition function and using a generalized form of Poisson summation.
We then demonstrated that the same expansion results from a somewhat non-standard supersymmetric localization calculation. 

A surprising feature of the latter calculation was the emergence of an effective action for a scalar mode which takes the form of a distribution, rather than a function.
The distributions we recover are specific regularizations of the singular functions which one obtains in the other localization approaches reviewed below.
Consequently, we expect that the results of any calculation which can be reliably carried out in either approach will agree.
It would be interesting to investigate how such a distribution could come about from a determinant over the fluctuations for a positive semi-definite localizing term, or the associated index theorem.
Preliminary steps in this direction were taken in \cite{Leeb-Lundberg:2023jsj}.

Benini and Zaffaroni have previously performed standard localization computations for supersymmetric gauge theories, which are further applicable to any A-twisted $\mathcal{N}=(2,2)$ theory on a closed Euclidean 2-manifold \cite{Benini:2015noa,Benini:2016hjo}.
Results for 2dYM can be straightforwardly read off from their work.
The effective action for the scalar modulus, which is complex in their setting, implied by their work is a meromorphic function, rather than a distribution.
In fact, meromorphic functions are generically produced by localization in this class of models, even in the presence of matter superfields \cite{Benini:2016hjo}.
Such singular effective actions would naively imply ambiguous results for protected observables.
However, detailed examination has shown that they can be dealt with in some cases.

A detailed examination of the breakdown of localization, or rather deformation invariance, in this context, specifically at genus one, was initiated in \cite{Benini:2013nda,Benini:2013xpa}.
The authors showed that, at least in a wide class of models satisfying some conditions on the gauge charges of chiral multiplets, the integration over the supersymmetric moduli should be carried out over a contour in the complex plane associated with the JK residue prescription \cite{jeffrey1993localization}.
This prescription, which was introduced following Witten's work on non-abelian localization, yields non-ambiguous results for protected observables.
The prescription has passed numerous checks.
The results were generalized to higher genus in \cite{Benini:2016hjo}. 

We do not believe that 2dYM falls within the class of models to which the JK residue prescription applies, but we cannot state with certainty that the results of \cite{Benini:2016hjo} cannot be generalized to
this case.
However, the analogous analysis of Chern--Simons theory, carried out in \cite{Benini:2015noa}, emphasizes the need to perform the sum over instanton sectors \emph{before} applying the JK residue prescription.
We believe that this implies that one would not be able to use the machinery of the JK residue to compute the perturbation series, rather only the final result, for 2dYM. 

The expression \eqref{eq:quantum_effective_action} is similar to the effective action for $\phi_{0}$ provided by the derivation using torus gauge in \cite{Blau:1993hj}, or from the twisted $\mathcal{N}=(2,2)$ perspective examined in \cite{Benini:2016hjo,Benini:2015noa}.
Both of these approaches yield an effective action of the form
\begin{equation}\label{eq:previous_effective_action}
  \phi_{0}^{2-2g} \;,
\end{equation}
although in \cite{Benini:2016hjo,Benini:2015noa} the status of this expression is different because the contour of integration for $\phi_{0}$ is not the same.
As explained in the toy example, \eqref{eq:quantum_effective_action} can be thought of as a specific regularization of \eqref{eq:previous_effective_action}.

There are a number of unsatisfactory aspects of the JK residue prescription which one may hope to address with an alternative localization procedure such as the one used here.
For one, to the best of our knowledge, the JK residue prescription can only be used for gauge theories with charged matter satisfying certain criteria and yielding specific types of singularities for the effective action. We do not believe that any such restriction applies to our procedure.
It is often useful to be able to gauge global symmetries at the level of the exact partition
function, without specifying the dynamical fields.
Such a process is possible in most localization computations by weakly gauging the symmetry, whereby the partition function depends on additional parameters.
However, the JK residue prescription requires full knowledge of the charges of dynamical chiral multiplets, and therefore presumably does not allow one to gauge global symmetries without reverting to the original description of the theory.
It seems plausible that this difficulty can be addressed in our procedure, since multiplication of distributions with compatible wave front sets is a well-defined process.

A JK residue type prescription is needed in order to evaluate analogous partition functions in higher dimensions, e.g.\ for twisted $\mathcal{N}=2$ theories on toric K\"{a}hler manifolds in four dimensions \cite{Bershtein:2015xfa} and the related $\mathcal{N}=1$ indices in five dimensions \cite{Hosseini:2018uzp}.
Because of the need for a mode by mode analysis, justifying the use of a specific prescription becomes increasingly difficult in these contexts.
The derivation of a distributional effective action may be simpler in these situations.

In upcoming work, we incorporate non-simply-connected groups and discrete theta angles, and also apply our techniques to the expectation values of observables like the Wilson loop. We expect to be able to generalize our results to higher rank groups.

\section*{Acknowledgments}
LG and DS are supported, in part, by Istituto Nazionale di Fisica Nucleare (INFN) through the
“Gauge and String Theory” (GAST) research project.
The research of RP is funded, in part, by the Deutsche Forschungsgemeinschaft (DFG, German Research Foundation) -- Projektnummer 277101999 -- TRR 183 (project A03 and B01).
JP acknowledges financial support from the European Research Council (grant BHHQG-101040024). Views and opinions expressed are however those of the authors only and do not necessarily reflect those of the European Union or the European Research Council. Neither the European Union nor the granting authority can be held responsible for them.
IY is supported, in part, by UKRI consolidated grants ST/P000711/1 and ST/T000775/1.

\appendix

\section{Fourier transforms}\label{APP:Fourier}
In this appendix we compute the Fourier transform of the distribution $f(x)$ defined in \eqref{EQ:f}. In order to do so, we consider the Fourier transforms of the two factors in \eqref{EQ:f}, which are given by
\allowdisplaybreaks
\par\vspace{10pt}
{\centering
\begin{tabular}{ccl}
\toprule
  $F(x)$ & $\hat{F}(p)$ & \\
\midrule
  $\displaystyle -\frac{1 }{(2 g-3)!} \, \frac{\d^{2g-3}}{\d x^{2 g-3}} \left(\pv\frac{1}{x}\right)$ \qquad\qquad & $\displaystyle\frac{(-\pi^2)^{g-1}}{(2g-3)!} \, (2|p|)^{2g-3}$ & for $g\in\mathbb{Z}$, $g\geq2$, \\[12pt]
  $\displaystyle e^{-\g x^2}$ & $\displaystyle\frac{\sqrt{\pi}}{a} \, e^{-\pi^2p^2/\g}$ & for $\g>0$. \\[5pt]
\bottomrule
\end{tabular}
\par}
\vspace{10pt}
\noindent
We can reconstruct the Fourier transform of \eqref{EQ:f} as the convolution of the two transforms in the table above, namely
\begin{align}
  \hat{f}(p)
  &= \frac{(-4\pi^2)^{g-1}}{(2g-3)!} \frac{\sqrt{\pi}}{2\sqrt{\g}} \int_{-\infty}^{\infty} \d q \; |q|^{2g-3} \, e^{-\pi^2(p-q)^2/\g} \cr
  &= \frac{(-4\g)^{g-1}}{(2g-3)!} \frac{\sqrt{\pi}}{2\sqrt{\g}} \int_{0}^{\infty} \d q \; q^{2g-3} \left( e^{-(q+\pi p/\sqrt{\g})^2} + e^{-(q-\pi p/\sqrt{\g})^2} \right) \;,
\end{align}
where in the last line we have split the integration range and changed integration variable with $q\mapsto\pm\pi q/\sqrt{\g}$.
The integral can be readily evaluated and expressed in terms of a Kummer confluent hypergeometric function with
\begin{align}\label{EQ:fhat_as_Kummer}
  \hat{f}(p)
  &= \frac{(-1)^{g-1} \pi}{\Gamma(g-1/2)} \, \g^{g-3/2} \;{}_1F_1(3/2-g;1/2;-\pi^2p^2/\g) \;.
\end{align}
However, for our purposes, it is useful to recast it as
\begin{align}
  \hat{f}(p) = \frac{(-4\g)^{g-1}}{(2g-3)!} \frac{\sqrt{\pi}}{2\sqrt{\g}} \, \Big( G(\pi p/\sqrt{\g}) + G(-\pi p/\sqrt{\g}) \Big) \;,
\end{align}
where we have introduced
\begin{align}
  G(x)
  &= \int_{0}^{\infty} \d q \; q^{2g-3} e^{-(q-x)^2} \cr
  &= \frac{e^{-x^2}}{2^{2g-3}} \, \frac{\partial^{2g-3}}{\partial x^{2g-3}} \int_{0}^{\infty} \d q \; e^{-q^2+2xq} \cr
  &= \frac{\sqrt{\pi} \, e^{-x^2}}{2^{2g-2}} \, \frac{\partial^{2g-3}}{\partial x^{2g-3}} \, \Big[e^{x^2} (1+\erf(x)) \Big] \;.
\end{align}
The above can be evaluated in terms of Hermite polynomials by using
\begin{align}
  \frac{\partial^k}{\partial x^k} \, e^{-x^2} = (-1)^k \, e^{-x^2} H_k(x) \;.
\end{align}

We can now rewrite \eqref{EQ:fhat_as_Kummer} as
\begin{align}
  \hat{f}(p) = e^{-\pi^2p^2/\g} \, P(p) + \erf(\pi p/\sqrt{\g}) \; Q'(p) \;,
\end{align}
where $P$ and $Q$ are even polynomials defined with
\begin{align}
  P(p) \label{EQ:P}
  &= (-1)^{g-1} \, (2\pi)^{2g-3} \, \sqrt{\frac{\g}{\pi}} \, p^{2g-4} \sum_{k=0}^{g-2}\bigg({-}\frac{\g}{2\pi^2p^2}\bigg)^{\!k}\sum_{m=0}^k \frac{(-1)^m\,(2k-2m-1)!!}{2^m\,m!\,(2g-2m-3)!} \;, \\
  Q(p) \label{EQ:Q}
  &= \frac{(-1)^{g-1}\sqrt{\pi}\,\g^{g-1}}{2\,\Gamma(g-1/2)} \; L_{g-1}^{-\frac{1}{2}}\bigg({-}\frac{\pi^2p^2}{\g}\bigg) \;.
\end{align}
The first definition can be obtained from
\begin{align}
  P(p)
  = (-1)^{g-1} \, 2\sqrt{\pi} \, \g^{g-\frac{3}{2}} \sum_{k=0}^{2g-4} \frac{\i^k}{k!\,(2g-3-k)!} \, H_k\bigg(\frac{\i\pi p}{\sqrt{\g}}\bigg) \, H_{2g-4-k}\bigg(\frac{\pi p}{\sqrt{\g}}\bigg) \;,
\end{align}
while the second implies\footnote{
  The second equality follows from the well-known relation between generalized Laguerre and Hermite polynomials
  \[
  H_{2n+1}(x) = (-1)^n \, 2^{2n+1} \, n! \, x \, L_n^{\frac{1}{2}}(x^2) \;.
  \]
}
\begin{align}\label{EQ:Q'}
  Q'(p)
  &= \frac{(-1)^{g-1}\pi^{\frac{5}{2}}\,\g^{g-2}}{\Gamma(g-1/2)} \, p \, L_{g-2}^{\frac{1}{2}}\bigg({-}\frac{\pi^2p^2}{\g}\bigg) \cr
  &= \frac{\i\pi\,\g^{g-\frac{3}{2}}}{(2g-3)!} \, H_{2g-3}\bigg(\frac{\i\pi p}{\sqrt{\g}}\bigg) \;.
\end{align}

\section{Yang--Mills connections}\label{sec:Yang-Mills-connections}

We will need several results about the moduli space of Yang--Mills connections on Riemann surfaces which can be found in \cite{Atiyah:1982fa}, or can be straightforwardly derived from the results therein.
We refer the reader to the original work of Atiyah and Bott \cite{Atiyah:1982fa} for a detailed analysis.

Let $\Sigma$ be a smooth two manifold with a Riemannian metric, and let $\star$ be the Hodge star operator.
The exterior derivative is denoted by $\d$, and we denote by $\d^{\dagger} \equiv -\star\d\,\star$ its adjoint with respect to the usual inner product for differential forms, namely,
\begin{align}
  (\eta,\omega) &\equiv \int_{\Sigma}\eta\wedge\star\omega \;.
\end{align}
Let $\eta$ be a differential form of degree $[\eta]$, then
\begin{align}
  \star^{2}\eta = (-1)^{[\eta]} \;.
\end{align}
Specifically, $\star^{2} = -1$ when acting on one-forms.
The Hodge star therefore splits the space of complexified one-forms $\Omega_{\mathbb{C}}^{1}$ into the forms with $\star$-eigenvalue $\pm\i$, denoted $\Omega^{1,0}$ and $\Omega^{0,1}$, and splits the exterior derivative $\d$ into $\d'$ and $\d''$ according to the image. This induces a holomorphic structure on $\Sigma$. 

Let $A$ be a connection on a principle $G$-bundle $P$ over a Riemann surface $\Sigma$, and let $F$ be its curvature.
$A$ satisfies the Yang--Mills equations if 
\begin{align}  
  \d_{A}^{\dagger}F=0\,,
\end{align}
where $\d_{A}$ is the covariant exterior derivative, and $\d_{A}^{\dagger}$ its adjoint.
The background gauge condition is defined to be $\d_{A}^{\dagger}A=0$.
The following facts apply to such an $A$ \cite{Atiyah:1982fa}.

The adjoint scalar $\star F$ is covariantly constant.
The eigenvalues of the operator
\begin{align}
  \Lambda \equiv \i[\star F,\cdot]
\end{align}
acting on the space of adjoint-valued forms are locally constant, and such forms can be decomposed accordingly. 
The maps given by the operator $\d_{A}$, and by its Hodge dual $\d_{A}^{\dagger}$, commute with $\Lambda$ and respect this decomposition. We define also the
operator
\begin{align}
  \hat{F} \equiv -\i\star\Lambda \;,
\end{align}
acting on adjoint-valued one-forms.

The covariant Laplacian is given by
\begin{align}
  \Delta_{A}^{\vphantom{\dagger}}
  \equiv \d_{A}^{\vphantom{\dagger}}\d_{A}^{\dagger} + \d_{A}^{\dagger}\d_{A}^{\vphantom{\dagger}} \;.
\end{align}
$\Delta_{A}$ is a positive semi-definite operator.
The tangent space to the moduli space of solutions at $A$ consists of adjoint-valued one-forms annihilated by both $\hat{F}$ and $\Delta_{A}$.
The Hessian for the Yang--Mills term, around a solution of the Yang--Mills equations and in background gauge, is equivalent to $\Delta_{A} + \hat{F}$ and vanishes only when acting on directions tangent to the moduli space.

The holomorphic structure defined on the tangent bundle of $\Sigma$ extends to the associated adjoint bundle $\operatorname{ad}(P)$.
This allows complexified adjoint-valued one-forms to be decomposed into $\Omega^{1,0}(\Sigma,\operatorname{ad}_{\mathbb{C}}(P))$ and $\Omega^{0,1}(\Sigma,\operatorname{ad}_{\mathbb{C}}(P))$ on which $\star$ acts as $-\i$ and $\i$ respectively, with a corresponding decomposition of the covariant exterior derivative into $\d_{A}'$ and $\d_{A}''$.

The holomorphic structure and the operators $\d_{A}''$ and $\d_{A}'$ are compatible with the decomposition according to eigenvalues of $\Lambda$.
We denote by $\operatorname{ad}_{\lambda}(P)$ the vector bundle with $\Lambda$ eigenvalue $\lambda$ and define the Laplacians 
\begin{align}
  \Box_{A}'' &\equiv \d_{A}''(\d_{A}'')^{\dagger} + (\d_{A}'')^{\dagger}\d_{A}'' \;, \\
  \Box_{A}' &\equiv \d_{A}'(\d_{A}')^{\dagger} + (\d_{A}')^{\dagger}\d_{A}' \;.
\end{align}
$\Box_{A}''$ and $\Box_{A}'$ define the same operator on $\Omega^{1,0}(\Sigma,\operatorname{ad}_{\lambda}(P))$ and $\Omega^{0,1}(\Sigma,\operatorname{ad}_{\lambda}(P))$, and when acting there,
\begin{align}
  \Delta_{A}^{\vphantom{'}} = 2\Box_{A}'' = 2\Box_{A}' \;.
\end{align}
When acting on $\Omega^{0,0}(\Sigma,\operatorname{ad}_{\lambda}(P))$, we instead have 
\begin{align}
  \Box_{A}' + \Box_{A}'' &= \Delta_{A}^{\vphantom{'}} \;, \\
  \Box_{A}' - \Box_{A}'' &= \lambda \;.
\end{align}
We have the following isomorphisms, following from Hodge theory and Serre duality,
\begin{align}
  \operatorname{ker}_{\operatorname{ad}_{\lambda}(P)\otimes\Omega^{1}}(\Delta_{A})
  &\simeq H_{\d_{A}''}^{1}(\Sigma,\operatorname{ad}_{\lambda}(P))\oplus H_{\d_{A}''}^{1}(\Sigma,\operatorname{ad}_{-\lambda}(P)) \;,\\
  \operatorname{ker}_{\operatorname{ad}_{\lambda}(P)}(\d_{A}'')
  &\simeq H_{\d_{A}''}^{0}(\Sigma,\operatorname{ad}_{\lambda}(P)) \;,\\
  \operatorname{ker}_{\operatorname{ad}_{\lambda}(P)}(\d_{A}')
  &\simeq\operatorname{ker}_{\operatorname{ad}_{\lambda}(P)\otimes\Omega^{2}}((\d_{A}'')^{\dagger})\simeq H_{\d_{A}''}^{0}(\Sigma,\operatorname{ad}_{-\lambda}(P)) \;.
\end{align}
Specifically,
\begin{align}\label{eq:harmonic_one_forms}
  \dim\operatorname{ker}_{\operatorname{ad}_{\lambda}(P)\otimes\Omega^{1}}(\Delta_{A})=\dim H_{\d_{A}''}^{1}(\Sigma,\operatorname{ad}_{\lambda}(P))+\dim H_{\d_{A}''}^{1}(\Sigma,\operatorname{ad}_{-\lambda}(P))\,.
\end{align}

For adjoint-valued one-forms with $\Lambda$ eigenvalue $\lambda$, the smallest non-zero eigenvalue of $\Delta_{A}$ is greater than or equal to $2|\lambda|$.
We will need the following result, which we derive from \cite{Atiyah:1982fa}:

\begin{proposition}
Let $v\in\Omega^{1,0}(\Sigma,\operatorname{ad}_{\mathbb{C}}(P))$ be such that $\Delta_{A}v=2|\lambda|v$. Then $v=\d_{A}'x$ for some $x\in\Omega^{0,0}(\Sigma,\operatorname{ad}_{\mathbb{C}}(P))$ satisfying $\d_{A}''x=0$.
Likewise, let $v\in\Omega^{0,1}(\Sigma,\operatorname{ad}_{\mathbb{C}}(P))$ be such that $\Delta_{A}v=2|\lambda|v$. Then $v=\d_{A}''x$ for some $x\in\Omega^{0,0}(\Sigma,\operatorname{ad}_{\mathbb{C}}(P))$ satisfying $\d_{A}'x=0$.
\end{proposition}
\begin{proof}
Let $v\in\Omega^{1}(\Sigma,\operatorname{ad}_{\lambda}(P))$ be such that $\Delta_{A}v=2|\lambda|v$. Assume $v\in\Omega^{0,1}(\Sigma,\operatorname{ad}_{\mathbb{C}}(P))$ and $\lambda<0$ so that $\hat{F}v = -|\lambda|v$. From the Hodge decomposition with respect to $\d_{A}''$ we have $v = \d_{A}''x$ since $v$ is not harmonic.
We always have 
\begin{align}
  \Box_{A}'x - \Box_{A}''x = \lambda x \;.
\end{align}
Now we compute
\begin{align*}
-2\lambda \d_{A}''x & =2\{\d_{A}'',(\d_{A}'')^{\dagger}\} \d_{A}''x\\
  & =2\d_{A}''\Box_{A}''x \;.
\end{align*}
Hence, by acting on the left with $(\d_{A}'')^{\dagger}$ we have
\begin{align}
  \Box_{A}''x = -\lambda x \;,
\end{align}
leading to $\Box_{A}'x = 0$, which implies
\begin{align}
  \d_{A}'x = 0 \;.
\end{align}
The case with $v\in\Omega^{1,0}(\Sigma,\operatorname{ad}_{\mathbb{C}}(P))$ and $\lambda>0$ is proved using the decomposition with respect to $\d_{A}'$.
\end{proof}

Combined with the isomorphisms above, this gives
\begin{equation}\label{eq:minimal_eigenvalue_one_forms}
  \dim(\{v\in\Omega^{1}(\Sigma,\operatorname{ad}_{\lambda}(P))\,|\,\Delta_{A}v=2|\lambda|v\} )=\dim H_{\d_{A}''}^{0}(\Sigma,\operatorname{ad}_{\lambda}(P))+\dim H_{\d_{A}''}^{0}(\Sigma,\operatorname{ad}_{-\lambda}(P)) \;.
\end{equation}
Finally, the Riemann--Roch theorem for $\d_{A}''$ states 
\begin{align}
  \dim H_{\d_{A}''}^{0}(\Sigma,\operatorname{ad}_{\lambda}(P))-\dim H_{\d_{A}''}^{1}(\Sigma,\operatorname{ad}_{\lambda}(P))=(c_{1}(\operatorname{ad}_{\lambda})+(1-g)\dim_{\mathbb{C}}\operatorname{ad}_{\lambda}) \;.
\end{align}

\bibliographystyle{JHEP}
\bibliography{bibliography}
\end{document}